\documentclass[11pt]{myclass}

\usepackage{mystyle}
\usepackage{subfigure}
\usepackage{mathtools}

\title{Geometric Inference on Kernel Density Estimates}
\author{
Jeff M. Phillips%
     \thanks{Thanks to supported by NSF CCF-1350888, IIS-1251019, and ACI-1443046.}
     \\ {\small\textsl{jeffp@cs.utah.edu}} \\  {\small University of Utah} \and 
Bei Wang%
      \thanks{Thanks to supported by  INL 00115847 via DOE DE-AC07ID14517, 
                                                            DOE NETL DEEE0004449, 
                                                            DOE DEFC0206ER25781, 
                                                            DOE DE-SC0007446, and
                                                            NSF 0904631.}
      \\ {\small\textsl{beiwang@sci.utah.edu}} \\  {\small University of Utah} \and
Yan Zheng \\ {\small\textsl{yanzheng@cs.utah.edu}} \\  {\small University of Utah}
}

\begin{document}

\begin{titlepage}
\maketitle 
\begin{abstract}
We show that geometric inference of a point cloud can be calculated by examining its kernel density estimate with a Gaussian kernel.  This allows one to consider kernel density estimates, which are robust to spatial noise, subsampling, and approximate computation in comparison to raw point sets.  This is achieved by examining the sublevel sets of the \emph{kernel distance}, which isomorphically map to superlevel sets of the kernel density estimate.  We prove new properties about the kernel distance, demonstrating stability results and allowing it to inherit reconstruction results from recent advances in distance-based topological reconstruction.  Moreover, we provide an algorithm to estimate its topology using weighted Vietoris-Rips complexes.
\end{abstract}
\end{titlepage}

\section{Introduction}
\label{sec:intro}

Geometry and topology have become essential tools in modern data analysis: geometry to handle spatial noise and topology to identify the core structure.  Topological data analysis (TDA) has found applications spanning protein structure analysis~\cite{EFFL95,LEFSS98} to heart modeling~\cite{GCZQMA13} to leaf science~\cite{PSMHW11}, and is the central tool of identifying quantities like connectedness, cyclic structure, and intersections at various scales.  Yet it can suffer from spatial noise in data, particularly outliers.  

When analyzing point cloud data, classically these approaches consider $\alpha$-shapes~\cite{Edelsbrunner1993}, where each point is replaced with a ball of radius $\alpha$, and the union of these balls is analyzed.  More recently a distance function interpretation~\cite{CC12} has become more prevalent where the union of $\alpha$-radius balls can be replaced by the sublevel set (at value $\alpha$) of the Hausdorff distance to the point set.  Moreover, the theory can be extended to other distance functions to the point sets, including the \emph{distance-to-a-measure}~\cite{ChazalCohen-SteinerMerigot2011} which is more robust to noise.  

This has more recently led to statistical analysis of TDA.  These results show not only robustness in the function reconstruction, but also in the topology it implies about the underlying dataset.  This work often operates on persistence diagrams which summarize the persistence (difference in function values between appearance and disappearance) of all homological features in single diagram.  A variety of work has developed metrics on these diagrams and probability distributions over them~\cite{MMH11,TMMH14}, and robustness and confidence intervals on their landscapes~\cite{Bub14,FLRWBS14,CFLRSW13} (summarizing again the most dominant persistent features~\cite{CFLRW14}).  Much of this work is independent of the function and data from which the diagram is generated, but it is now more clear than ever, it is most appropriate when the underlying function is robust to noise, e.g., the distance-to-a-measure~\cite{ChazalCohen-SteinerMerigot2011}.  

\begin{figure}[h!]
\includegraphics[width=0.3\linewidth]{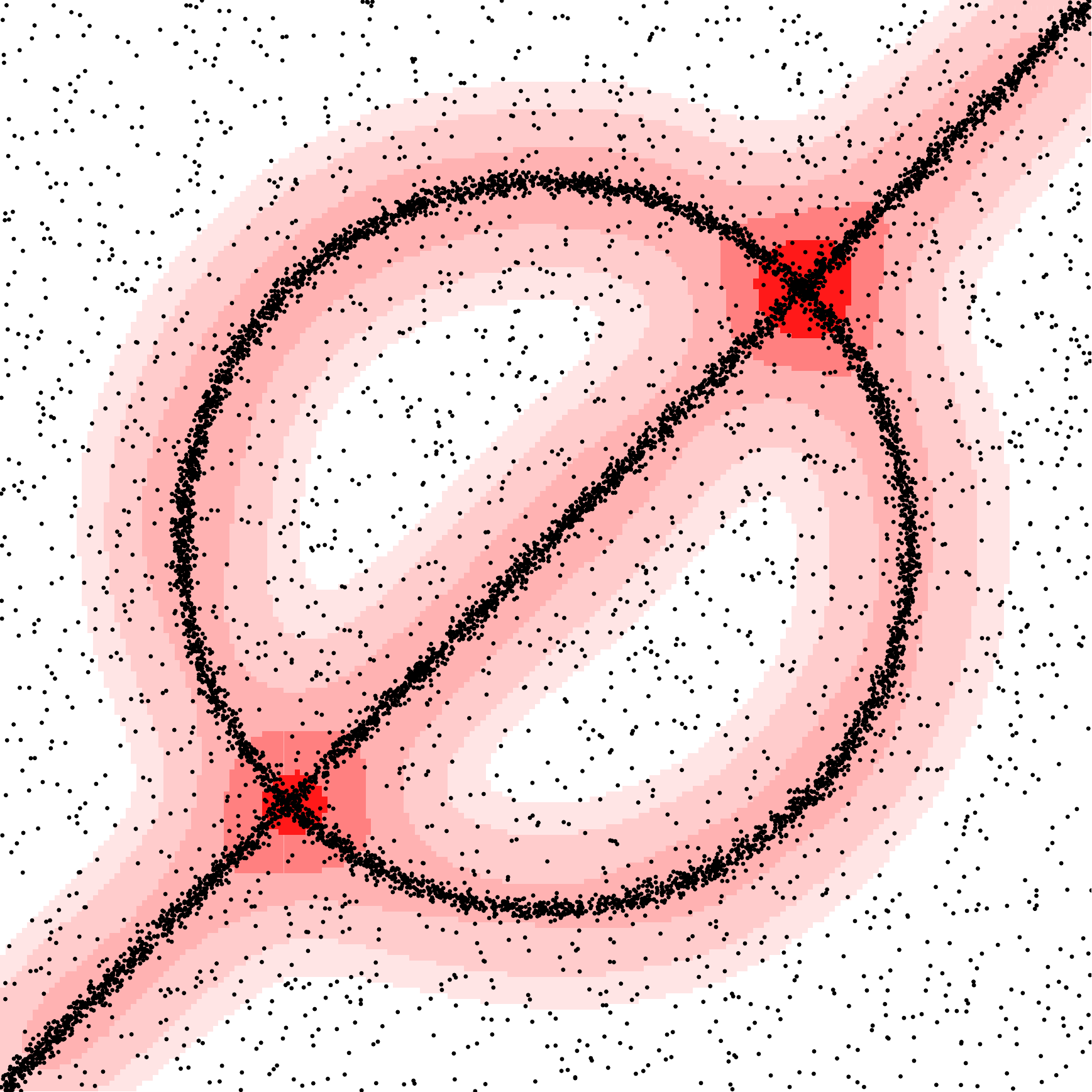}
\hspace{0.01\linewidth}
\includegraphics[width=0.32\linewidth]{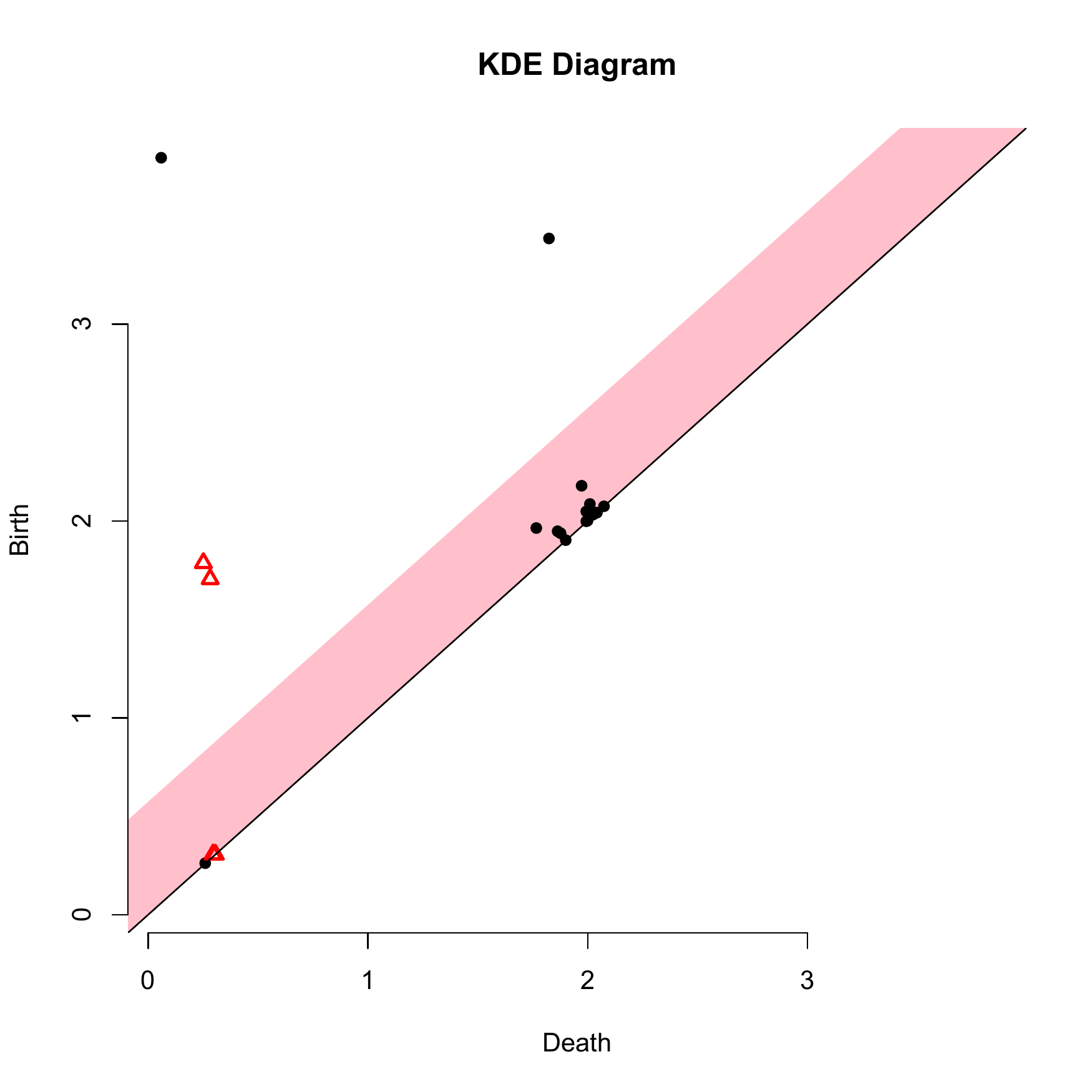}
\hspace{0.01\linewidth}
\includegraphics[width=0.32\linewidth]{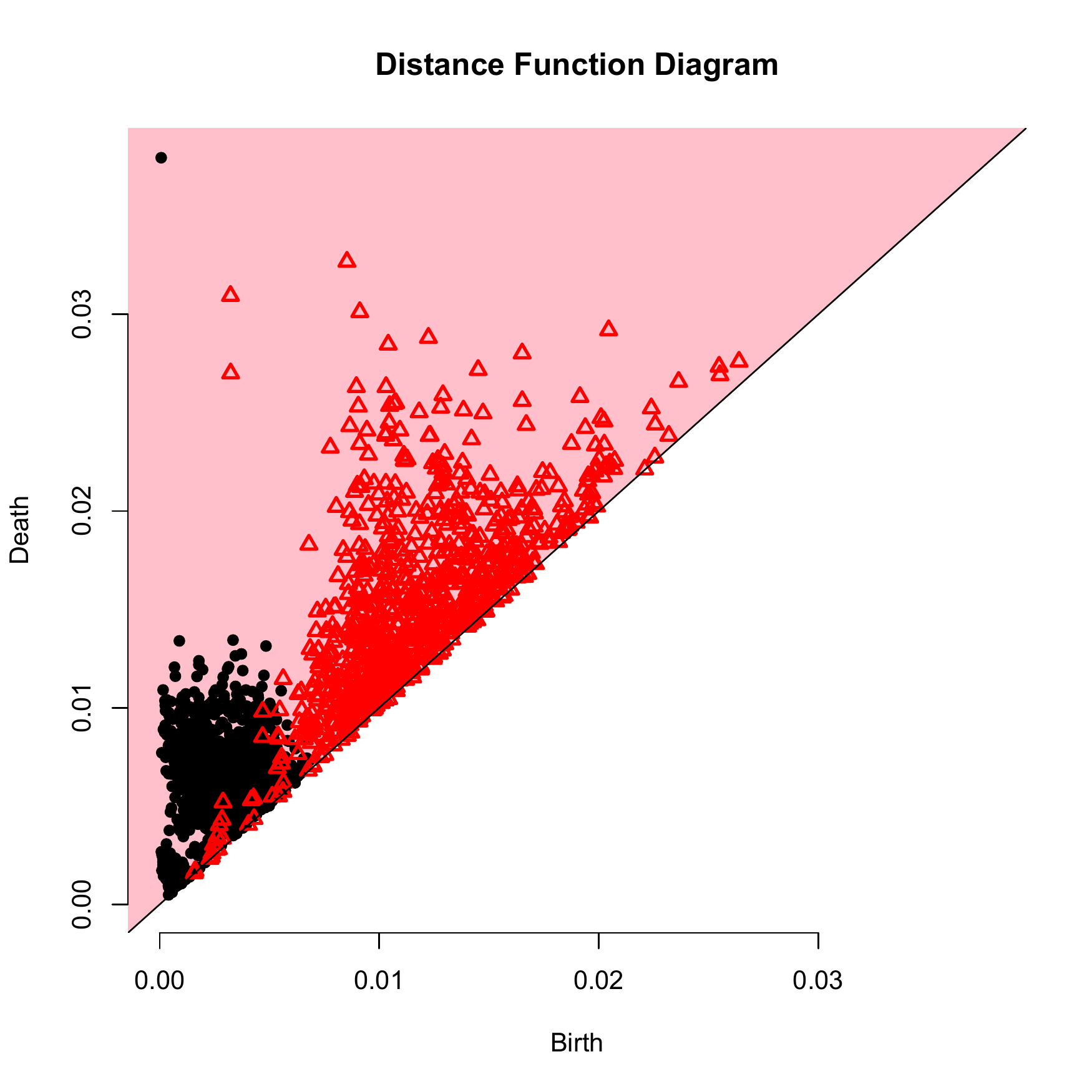}
\vspace{-.11in}
\\
\includegraphics[width=0.30\linewidth]{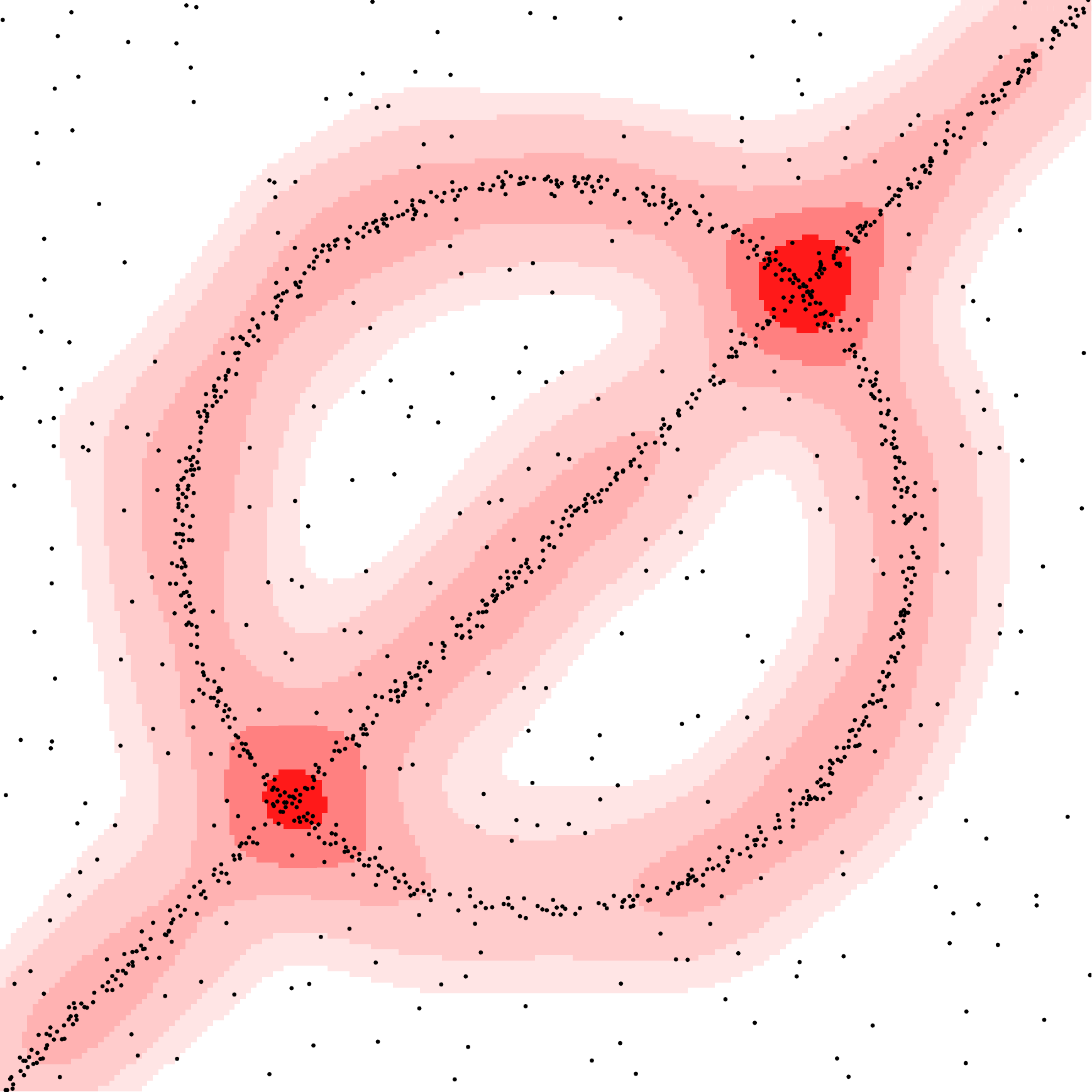}
\hspace{0.01\linewidth}
\includegraphics[width=0.32\linewidth]{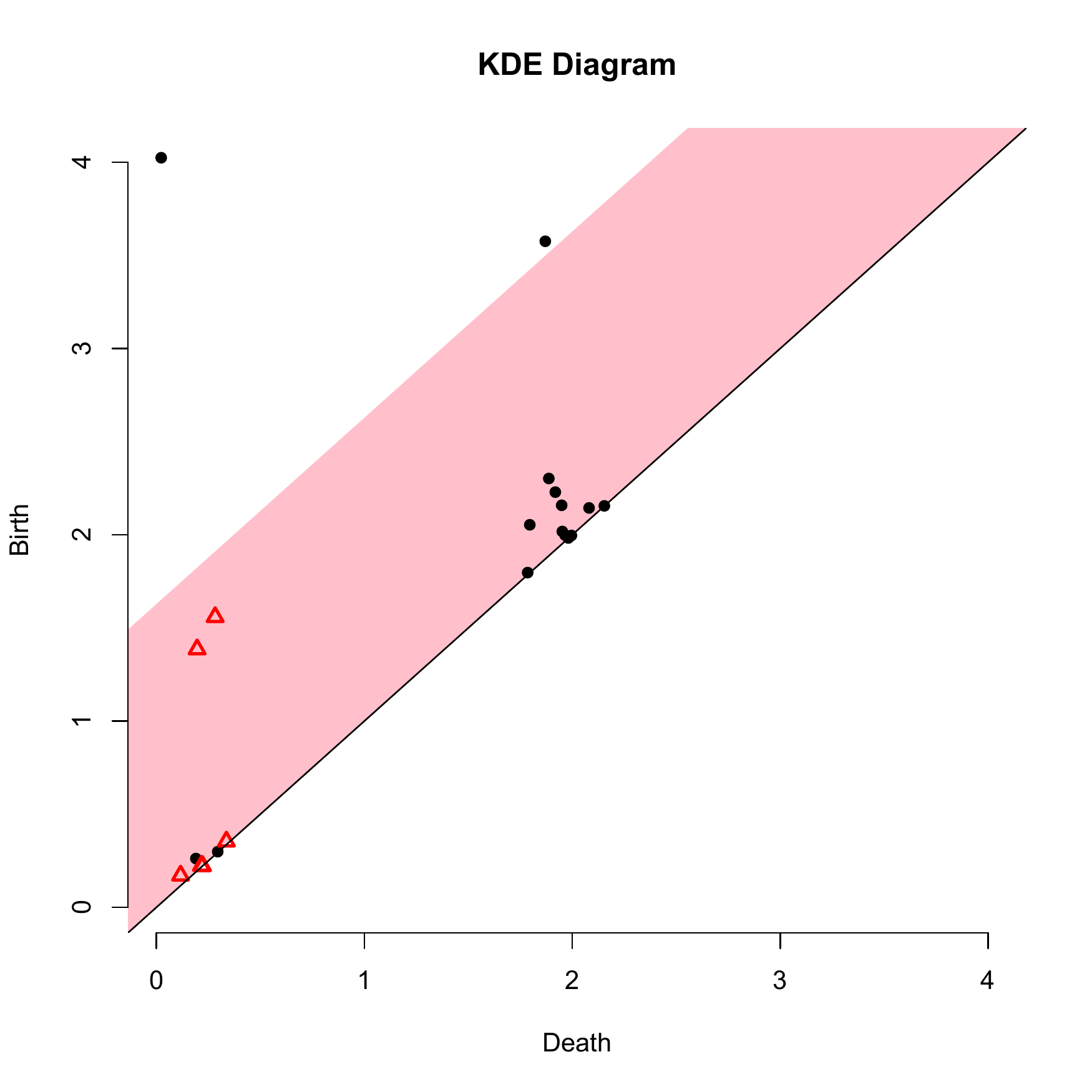}
\hspace{0.01\linewidth}
\includegraphics[width=0.32\linewidth]{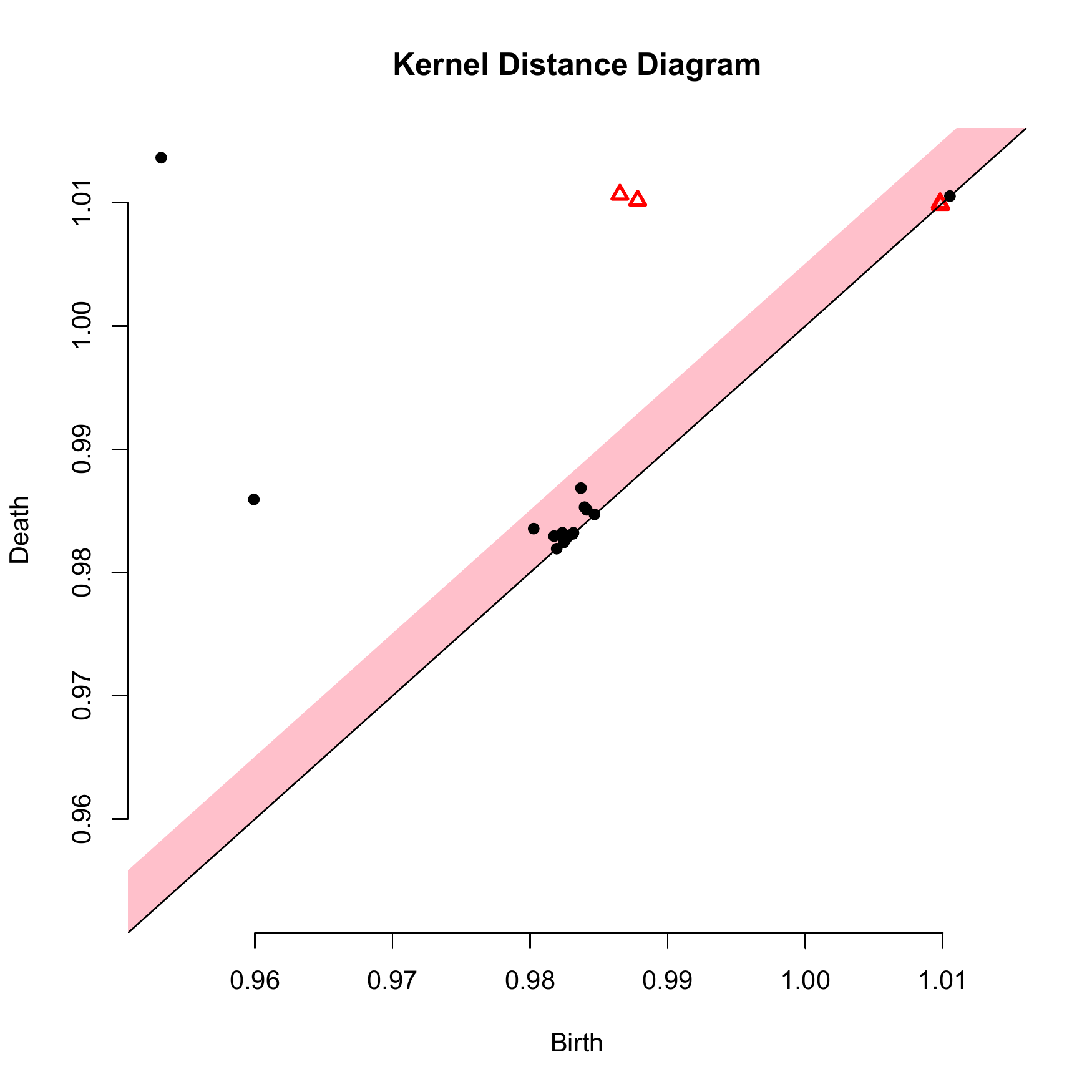}
\vspace{-.11in}
\caption{\label{fig:circleline}
\small \sffamily 
Example with $10{,}000$ points in $[0,1]^2$ generated on a circle or line with $N(0,0.005)$ noise; $25\%$ of points are uniform background noise.  The generating function is reconstructed with \kde\ with $\sigma = 0.05$ (upper left), and its persistence diagram based on the superlevel set filtration is shown (upper middle).  
A coreset~\cite{big-kde} of the same dataset with only $1{,}384$ points (lower left) and persistence diagram (lower middle) are shown, again using \kde.  This associated confidence interval contains the dimension $1$ homology features (red triangles) suggesting they are noise; this is because it models data as iid -- but the coreset data is not iid, it subsamples more intelligently.   
We also show persistence diagrams of the original data based on the sublevel set filtration of the standard distance function (upper right, with no useful features due to noise) and the kernel distance (lower right).}
\end{figure}

A very recent addition to this progression is the new TDA package for R~\cite{FKLM14}; it includes built in functions to analyze point sets using Hausdorff distance, distance-to-a-measure, $k$-nearest neighbor density estimators, kernel density estimates, and kernel distance.  The example in Figure \ref{fig:circleline} used this package to generate persistence diagrams.  
While, the stability of the Hausdorff distance is classic~\cite{CC12,Edelsbrunner1993}, and the distance-to-a-measure~\cite{ChazalCohen-SteinerMerigot2011} and $k$-nearest neighbor distances have been shown robust to various degrees~\cite{BCCSDR11}, this paper is the first to analyze the stability of kernel density estimates and the kernel distance in the context of geometric inference.
Some recent manuscripts show related results.  
Bobrowski \etal~\cite{BMT14} consider kernels with finite support, and describe approximate confidence intervals on the superlevel sets, which recover approximate persistence diagrams.  
Chazal \etal~\cite{CFLMRW14} explore the robustness of the kernel distance in bootstrapping-based analysis.

In particular, we show that the kernel distance and kernel density estimates, using the Gaussian kernel, inherit some reconstruction properties of distance-to-a-measure, that these functions can also be approximately reconstructed using weighted (Vietoris-)Rips complexes~\cite{BuchetChazalOudot2013}, and that under certain regimes can infer homotopy of compact sets.  
Moreover, we show further robustness advantages of the kernel distance and kernel density estimates, including that they possess small coresets~\cite{Phillips2013,big-kde} for persistence diagrams and inference.  

\vspace{-.1in}
\subsection{
Kernels, Kernel Density Estimates, and Kernel Distance}
\label{sec:kernel}
A \emph{kernel} is a non-negative similarity measure $K : \Rspace^d \times \Rspace^d \to \Rspace^+$; more similar points have higher value.  For any fixed $p \in \Rspace^d$, a kernel $K(p,\cdot)$ can be normalized to be a probability distribution; that is $\int_{x \in \Rspace^d} K(p,x) \dir{x} = 1$.  
For the purposes of this article we focus on the Gaussian kernel defined as 
$K(p,x) =  \sigma^2 \exp(-\|p - x\|^2 / 2\sigma^2)$. \footnote{
$K(p,x)$ is normalized so that $K(x,x) = 1$ for $\sigma = 1$.  
The choice of coefficient $\sigma^2$ is not the standard normalization, but it is perfectly valid as it scales everything by a constant.  It has the property that $\sigma^2 - K(p,x) \approx \|p-x\|^2/2$ for $\|p-x\|$ small.}

A \emph{kernel density estimate}~\cite{Sil86,Sco92,DG84,DL01} is a way to estimate a continuous distribution function over $\Rspace^d$ for a finite point set $P \subset \Rspace^d$; they have been studied and applied in a variety of contexts, for instance, under subsampling~\cite{Phillips2013,big-kde,BFLRSW13}, motion planning~\cite{PEKK12}, multimodality~\cite{Sil81,EFR12}, and surveillance~\cite{EDHD02}, road reconstruction~\cite{BE12}.
Specifically, 
\[
\kde_P(x) = \frac{1}{|P|} \sum_{p \in P} K(p,x). 
\]

The \emph{kernel distance}~\cite{HB05,glaunesthesis,
JoshiKommarajuPhillips2011,PhillipsVenkatasubramanian2011} (also called \emph{current distance} or \emph{maximum mean discrepancy}) is a metric~\cite{Muller1997,SGFSL10} between two point sets $P$, $Q$ (as long as the kernel used is characteristic~\cite{SGFSL10}, a slight restriction of being positive definite~\cite{Aronszajn1950,Wah99}, this includes the Gaussian and Laplace kernels).
Define a similarity between the two point sets as 
\[
\kappa(P,Q) = \frac{1}{|P|}\frac{1}{|Q|} \sum_{p \in P} \sum_{q \in Q} K(p,q).
\] 
Then the kernel distance between two point sets is defined as \vspace{-1mm}
\[
D_K(P,Q) = \sqrt{\kappa(P,P) + \kappa(Q,Q) - 2 \kappa(P,Q)}.
\]
When we let point set $Q$ be a single point $x$, then $\kappa(P,x) = \kde_P(x)$.  

Kernel density estimates can apply to any measure $\mu$ (on $\Rspace^d$) as 
$
\kde_\mu(x) = \int_{p \in \Rspace^d} K(p,x) \dir{\mu(p)}.  
$ 
The similarity between two measures is 
$
\kappa(\mu,\nu) = \int_{(p,q) \in \Rspace^d \times \Rspace^d} K(p,q) \dir{\textsf{m}_{\mu,\nu}(p,q)}, 
$ 
where $\textsf{m}_{\mu,\nu}$ is the product measure of $\mu$ and $\nu$ ($\textsf{m}_{\mu,\nu}:= \mu \otimes \nu$), and then the kernel distance between two measures $\mu$ and $\nu$ is still a metric, defined as  
$
D_K(\mu,\nu) = \sqrt{\kappa(\mu,\mu) + \kappa(\nu,\nu) - 2 \kappa(\mu,\nu)}.
$ 
When the measure $\nu$ is a Dirac measure at $x$ ($\nu(q) = 0$ for $x \neq q$, but integrates to $1$), then $\kappa(\mu,x) = \kde_{\mu}(x)$.  
Given a finite point set $P \subset \Rspace^d$, we can work with the empirical measure $\mu_P$ defined as $\mu_P = \frac{1}{|P|} \sum_{p \in P} \delta_p$, where $\delta_p$ is the Dirac measure on $p$, and $D_K(\mu_P, \mu_Q) = D_K(P,Q)$.  

If $K$ is positive definite, it is said to have the reproducing property~\cite{Aronszajn1950,Wah99}.
This implies that $K(p,x)$ is an inner product in some reproducing kernel Hilbert space (RKHS) $\Eu{H}_K$.  Specifically, there is a lifting map $\phi : \Rspace^d \to \Eu{H}_K$ so that $K(p,x) = \langle \phi(p), \phi(x) \rangle_{\Eu{H}_K}$, and moreover the entire set $P$ can be represented as $\Phi(P) = \sum_{p \in P} \phi(p)$, which is a single element of $\Eu{H}_K$ and has a norm $\|\Phi(P)\|_{\Eu{H}_K} = \sqrt{\kappa(P,P)}$.  A single point $x \in \Rspace^d$ also has a norm $\|\phi(x)\|_{\Eu{H}_K} = \sqrt{K(x,x)}$ in this space.


\subsection{Geometric Inference and Distance to a Measure: A Review}
\label{sec:ball}

Given an unknown compact set $S \subset \Rspace^d$ and a finite point cloud $P \subset \Rspace^d$ that comes from $S$ under some process, geometric inference aims to recover topological and geometric properties of $S$ from $P$. The offset-based (and more generally, the distance function-based) approach for geometric inference reconstructs a geometric and topological approximation of $S$ by offsets from $P$ (e.g. \cite{ChazalCohen-SteinerLieutier2009b,ChazalCohen-SteinerLieutier2009,ChazalCohen-SteinerMerigot2011,ChazalLieutier2005,ChazalLieutier2006}).

Given a compact set $S \subset \Rspace^d$, we can define a \emph{distance function} $f_S$ to  $S$; a common example is $f_S(x) = \inf_{y \in S} \|x - y\|$ (i.e. $\alpha$-shapes).
The offsets of $S$ are the sublevel sets of $f_S$, denoted $(S)^r = f_S^{-1}([0,r])$.
Now an approximation of $S$ by another compact set $P \subset \Rspace^d$ (e.g. a finite point cloud) can be quantified by the Hausdorff distance $d_H(S, P):=\|f_S - f_P\|_{\infty} = \inf_{x\in \Rspace^d}|f_S(x) - f_{P}(x)|$ of their distance functions. 
The intuition behind the inference of topology is that if $d_H(S, P)$ is small, thus $f_S$ and $f_{P}$ are close, and subsequently, $S$, ${(S)}^r$ and ${(P)}^r$ carry the same topology for an appropriate scale $r$. 
In other words, to compare the topology of offsets ${(S)}^r$ and ${(P)}^r$, 
we require Hausdorff stability with respect to their distance functions $f_S$ and $f_P$.  
 An example of an offset-based topological inference result is formally stated as follows (as a particular version of the reconstruction Theorem 4.6 in \cite{ChazalCohen-SteinerLieutier2009}), 
 where the \emph{reach} of a compact set $S$, $\reach(S)$, is defined as the minimum distance between $S$ and its medial axis \cite{Merigot2010}. 

\begin{theorem}[Reconstruction from $f_P$ \cite{ChazalCohen-SteinerLieutier2009}]
\label{thm:recon-fP}
Let $S, P \subset \Rspace^d$ be compact sets such that $\reach(S) > R$ and $\eps:= d_H(S, P) < R/17$. Then $(S)^\eta$ and ${(P)}^{r}$ are homotopy equivalent for sufficiently small $\eta$ (e.g., $0 < \eta < R$) if $4\eps \leq r <  R-3\eps$.
\end{theorem}

Here $\eta < R$ ensures that the topological properties of $(S)^\eta$ and $(S)^r$ are the same, and the $\eps$ parameter ensures $(S)^r$ and $(P)^r$ are close.  
Typically $\eps$ is tied to the density with which a point cloud $P$ is sampled from $S$.  

For function $\phi : \mathbb{R}^d \to \mathbb{R}^+$ to be \emph{distance-like} it should satisfy the following properties:
\begin{itemize} \denselist
\item (D1)  $\phi$ is $1$-Lipschitz:  For all $x,y \in \Rspace^d$, $|\phi(x) - \phi(y)| \leq \|x-y\|$.  
\item (D2)  $\phi^2$ is $1$-semiconcave:  The map $x \in \Rspace^d \mapsto (\phi(x))^2 - \|x\|^2$ is concave.  
\item (D3)   $\phi$ is proper: $\phi(x)$ tends to the infimum of its domain (e.g., $\infty$) as $x$ tends to infinity.  
\end{itemize}
In addition to the Hausdorff stability property stated above, as explained in \cite{ChazalCohen-SteinerMerigot2011}, $f_S$ is distance-like.  These three properties are paramount for geometric inference (e.g. \cite{ChazalCohen-SteinerLieutier2009,Lieutier2004}).  
(D1) ensures that $f_S$ is differentiable almost everywhere and the medial axis of $S$ has zero $d$-volume~\cite{ChazalCohen-SteinerMerigot2011}; and (D2) is a crucial technical tool, e.g., in proving the existence of the flow of the gradient of the distance function for topological inference \cite{ChazalCohen-SteinerLieutier2009}.

\paragraph{Distance to a measure.}
Given a probability measure $\mu$ on $\Rspace^d$ and a parameter $m_0 > 0$ smaller than the total mass of $\mu$,
the \emph{distance to a measure} $\dC_{\mu,m_0}: \Rspace^n \to \Rspace^+$ \cite{ChazalCohen-SteinerMerigot2011} is defined for any point $x \in \Rspace^d$ as 
\[
\dC_{\mu, m_0}(x) = \left(\frac{1}{m_0} \int_{m = 0}^{m_0} (\delta_{\mu,m}(x))^2 \dir{m}\right)^{1/2},
\;\; \text{ where } \;\;
\delta_{\mu,m}(x) = \inf \left\{r > 0 :  \mu(\bar{B}_r(x)) \geq m \right\},
\]
It has been shown in \cite{ChazalCohen-SteinerMerigot2011} that $\dC_{\mu, m_0}$ is a distance-like function (satisfying (D1), (D2), and (D3)), and:
\begin{itemize} 
\item (M4)   [Stability] For probability measures $\mu$ and $\nu$ on $\Rspace^d$ and $m_0>0$, then $\| \dC_{\mu, m_0} - \dC_{\nu,m_0}\|_{\infty} \leq \frac{1}{\sqrt{m_0}} W_2 (\mu, \nu)$.
\end{itemize}
Here $W_2$ is the \emph{Wasserstein distance}~\cite{Villani2003}:  
$
W_2(\mu, \nu) = \inf_{\pi \in \Pi(\mu,\nu)} \left( \int_{\Rspace^d \times \Rspace^d}||x-y||^2 \dir{\pi}(x,y) \right)^{1/2}
$
between two measures, where $\dir{\pi}(x,y)$ measures the amount of mass transferred from location $x$ to location $y$ and $\pi \in \Pi(\mu, \nu)$ is a transference plan \cite{Villani2003}.  


Given a point set $P$, the sublevel sets of $\dC_{\mu_P,m_0}$ can be described as the union of balls \cite{GuibasMerigotMorozov2011}, and then one can algorithmically estimate the topology (e.g., persistence diagram) with weighted alpha-shapes \cite{GuibasMerigotMorozov2011} and weighted Rips complexes \cite{BuchetChazalOudot2013}. 

\subsection{Our Results}
\label{sec:results}

We show how to estimate the topology (e.g., approximate persistence diagrams, infer homotopy of compact sets) using superlevel sets of the kernel density estimate of a point set $P$.  
We accomplish this by showing that a similar set of properties hold for the kernel distance with respect to a measure $\mu$,  
(in place of distance to a measure $\dC_{\mu, m_0}$),  defined as 
\[
\dK_\mu(x) = D_K(\mu,x) = \sqrt{\kappa(\mu,\mu) + \kappa(x,x) - 2\kappa(\mu,x)}.
\]
This treats $x$ as a probability measure represented by a Dirac mass at $x$.  
Specifically, we show $\dK_\mu$ is distance-like (it satisfies (D1), (D2), and (D3)), so it inherits reconstruction properties of $\dC_{\mu,m_0}$.  Moreover, it is stable with respect to the kernel distance:
\begin{itemize}
\item (K4) [Stability] If $\mu$ and $\nu$ are two measures on $\Rspace^d$, then 
$\| \dK_\mu - \dK_\nu \|_\infty \leq D_K(\mu,\nu)$. 
\end{itemize}

In addition, we show how to construct these topological estimates for $\dK_\mu$ using weighted Rips complexes, following power distance machinery introduced in \cite{BuchetChazalOudot2013}.  That is, a particular form of power distance permits a multiplicative approximation with the kernel distance.  

We also describe further advantages of the kernel distance.  
(i) Its sublevel sets conveniently map to the superlevel sets of a kernel density estimate.  
(ii) It is Lipschitz with respect to the smoothing parameter $\sigma$ when the input $x$ is fixed.  
(iii) As $\sigma$ tends to $\infty$ for any two probability measures $\mu,\nu$, the kernel distance is bounded by the Wasserstein distance:  $\lim_{\sigma \to \infty} D_K(\mu,\nu) \leq W_2(\mu,\nu)$.  
(iv) It has a small coreset representation, which allows for sparse representation and efficient, scalable computation.  In particular, an \emph{$\eps$-kernel sample}~\cite{JoshiKommarajuPhillips2011,Phillips2013,big-kde} $Q$ of $\mu$ is a finite point set whose size only depends on $\eps>0$ and such that $\max_{x \in \Rspace^d} |\kde_\mu(x) - \kde_{\mu_Q}(x) | = \max_{x \in \Rspace^d} |\kappa(\mu,x) - \kappa(\mu_Q,x) | \leq \eps$.  These coresets preserve inference results and persistence diagrams.

\section{Kernel Distance is Distance-Like} 
\label{sec:prop}

In this section we prove $\dK_\mu$ satisfies (D1), (D2), and (D3); hence it is distance-like.  
Recall we use the $\sigma^2$-normalized Gaussian kernel $K_\sigma(p,x) = \sigma^2 \exp(-\|p - x\|^2 / 2 \sigma^2)$.  For ease of exposition, unless otherwise noted, we will assume $\sigma$ is fixed and write $K$ instead of $K_\sigma$.  

\subsection{Semiconcave Property for $\dK_\mu$}

\begin{lemma}[D2]
$(\dK_\mu)^2$ is $1$-semiconcave: the map $x \mapsto (\dK_\mu(x))^2 - \|x\|^2$ is concave. 
\label{lem:semi}

\end{lemma}
\begin{proof}
Let $T(x) =  (\dK_\mu(x))^2 - \|x\|^2$.  The proof will show that the second derivative of $T$ along any direction is nonpositive.  
We can rewrite 
\begin{align*}
T(x) 
&= 
\kappa(\mu,\mu) + \kappa(x,x) - 2\kappa(\mu,x) - \|x\|^2 
\\ &= 
\kappa(\mu,\mu) + \kappa(x,x) - \int_{p \in \Rspace^d} (2K(p,x) +\|x\|^2) \dir{\mu(p)}.
\end{align*}
Note that both $\kappa(\mu,\mu)$ and $\kappa(x,x)$ are absolute constants, so we can ignore them in the second derivative.  
Furthermore, by setting $t(p,x) = -2K(p,x) - \|x\|^2$, the second derivative of $T(x)$ is nonpositive if the second derivative of $t(p,x)$ is nonpositive for all $p,x \in \Rspace^d$.  
First note that the second derivative of $-\|x\|^2$ is a constant $-2$ in every direction.  
The second derivative of $K(p,x)$ is symmetric about $p$, so we can consider the second derivative along any vector $u = x-p$,
\[
\frac{\dir{}^2}{\dir u^2} t(p,x) 
= 
2\left(\frac{ \|u\|^2}{\sigma^2} -1\right) \exp\left(-\frac{\| u \|^2}{2\sigma^2}\right) - 2.
\]
This reaches its maximum value at $\| u \| = \|x-p\| = \sqrt{3}\sigma$ where it is $4\exp(-3/2)-2 \approx -1.1$; this follows setting the derivative of $s(y) = 2(y-1) \exp(-y/2)-2$ to $0$, ($\frac{\dir{}}{\dir y} s(y) = (1/2)(3-y)\exp(-y/2)$), substituting $y = \|u\|^2/\sigma^2$.  
\end{proof}

We also note in Appendix \ref{app:dist-like} that semiconcavity follows trivially in the RKHS $\Eu{H}_K$.  

\subsection{Lipschitz Property for $\dK_\mu$}
\label{sec:lipschitz}

We generalize a (folklore, see \cite{ChazalCohen-SteinerMerigot2011}) relation between semiconcave and Lipschitz functions and prove it for completeness.  
A function $f$ is $\ell$-semiconcave if the function $T(x) = (f(x))^2 - \ell \|x\|^2$ is concave.  

\begin{lemma}
Consider a twice-differentiable function $g$ and a parameter $\ell \geq 1$.  
If $(g(x))^2$ is $\ell$-semiconcave, then $g(x)$ is $\ell$-Lipschitz.  
\label{lem:semi-Lip}
\end{lemma}
\begin{proof}
The proof is by contrapositive; we assume that $g(x)$ is not $\ell$-Lipschitz and then show $(g(x))^2$ cannot be $\ell$-semiconcave.  
By this assumption, then in some direction $u$, there is a point $x'$ such that $(d/du) g(x') = c > \ell \geq 1$.  

Now we examine $f(x) = (g(x))^2 - \ell \|x\|^2$ at $x = x'$, and specifically its second derivative in direction $u$. 
\begin{align*}
\frac{d}{du} f(x) \big|_{x=x'} &= 2 \left(\frac{d}{du} g(x')\right) g(x') - 2 \ell \|x'\| = 2 c \cdot g(x') - 2 \ell \|x'\|
\\
\frac{d^2}{du^2} f(x) \big|_{x=x'} &= 2 c \left(\frac{d}{du} g(x')\right) - 2 \ell = 2c^2 - 2 \ell = 2 (c^2 - \ell)
\end{align*}
Since $c^2 > c > \ell \geq 1$, then $2 (c^2 - \ell) > 0$ and $f(x)$ is not $\ell$-semiconcave at $x'$.  
\end{proof}

We can now state the following lemma as a corollary of Lemma \ref{lem:semi-Lip} and Lemma \ref{lem:semi}.  

\begin{lemma}[D1]
$\dK_\mu$ is $1$-Lipschitz on its input. 
\label{lem:Lipschitz}
\end{lemma}

\subsection{Properness of $d^K_{\mu}$}

Finally, for $\dK_\mu$ to be distance-like, we need to show it is proper when its range is restricted to be less than $c_\mu := \sqrt{\kappa(\mu, \mu) + \kappa(x,x)}$. Here, the value of $c_\mu$ depends on $\mu$ not on $x$ since $\kappa(x,x) = K(x,x) = \sigma^2$.
This is required for a distance-like version \cite{ChazalCohen-SteinerMerigot2011}, Proposition 4.2) of the Isotopy Lemma (\cite{Grove1993}, Proposition 1.8).  

\begin{lemma}[D3] 
$d^K_{\mu}$ is proper. 
\label{lem:proper}
\end{lemma}

We delay this technical proof to Appendix \ref{app:dist-like}.  
The main technical difficulty comes in mapping standard definitions and approaches for distance functions to our function $\dK_\mu$ with a restricted range.  

Also by properness (see discussion in Appendix \ref{app:dist-like}), Lemma \ref{lem:proper} also implies that $\dK_\mu$ is a closed map and its levelset at any value $a \in [0, c_\mu)$ is compact. 
This also means that the sublevel set of $d^K_{\mu}$ (for ranges $[0, a) \subset [0, c_\mu)$) is compact. Since the levelset (sublevel set) of $\dK_{\mu}$ corresponds to the levelset (superlevel set) of $\kde_{\mu}$, we have the following corollary.  
\begin{corollary}
The superlevel sets of $\kde_{\mu}$ for all ranges with threshold $a > 0$, are compact. 
\end{corollary}
The result in \cite{EFR12} shows that given a measure $\muP$ defined by a point set $P$ of size $n$, the $\kde_{\muP}$ has polynomial in $n$ modes; hence the superlevel sets of $\kde_{\muP}$ are compact in this setting.  The above corollary is a more general statement as it holds for any measure.

\section{Power Distance using Kernel Distance}
\label{sec:powerdistance}
A \emph{power distance} using $d^K_{\mu}$ is defined with a point set $P \subset \Rspace^d$ and a metric $d(\cdot, \cdot)$ on $\Rspace^d$,  
\[
\Dpow{P}(\mu, x) = \sqrt{\min_{p \in P}  \left( d(p,x)^2 + \dK_\mu(p)^2 \right)}.
\]
A point $x \in \mathbb{R}^d$ takes the distance under $d(p,x)$ to the closest $p \in P$, plus a weight from $\dK_\mu(p)$; thus a sublevel set of $\Dpow{P}(\mu, \cdot)$ is defined by a union of balls.  
We consider a particular choice of the distance $d(p,x) := D_K(p,x)$ which leads to a kernel version of power distance
\[
\Kpow{P}(\mu, x) = \sqrt{\min_{p \in P} \left( {D_K(p,x)}^2 + {\dK_\mu(p)}^2 \right) }.  
\]

In Section \ref{subsec:topo-est} we use $\Kpow{P}(\mu,x)$ to adapt the construction introduced in \cite{BuchetChazalOudot2013} to approximate the persistence diagram of the sublevel sets of $\dK_\mu$, using a weighted Rips filtration of $\Kpow{P}(\mu,x)$. 

Given a measure $\mu$, let $p_+ = \arg \max_{q \in \Rspace^d} \kappa(\mu,q)$, and let $P_+ \subset \Rspace^d$ be a point set that contains $p_+$.  We show below, in Theorem \ref{thm:powK-up} and Theorem \ref{thm:powK-low}, that $\frac{1}{\sqrt{2}} \dK_\mu(x) \leq \Kpow{P_+}(\mu,x) \leq \sqrt{14} \dK_\mu(x)$.  
However, constructing $p_+$ exactly seems quite difficult.  
We also attempt to use $p^\star = \arg\min_{p \in P} \|p-x\|$ in place of $p_+$ (see Section \ref{app:pow-justP}), but are not able to obtain useful bounds.  

Now consider an empirical measure $\mu_P$ defined by a point set $P$.  
We show (in Theorem \ref{thm:phat+} in Appendix \ref{app:phat+}) how to construct a point $\hat p_+$  (that approximates $p_+$) such that $D_K(P,\hat p_+) \leq (1+\delta) D_K(P, p_+)$ 
for any $\delta >0$.  For a point set $P$, the \emph{median concentration} $\Lambda_P$ is a radius such that no point $p \in P$ has more than half of the points of $P$ within $\Lambda_P$, 
and the \emph{spread} $\beta_P$ is the ratio between the longest and shortest pairwise distances.  
The runtime is polynomial in $n$ and $1/\delta$ assuming $\beta_P$ is bounded, and  that $\sigma/\Lambda_P$ and $d$ are constants.  

Then we consider $\hat P_+ = P \cup \{\hat p_+\}$, where $\hat p_+$ is found with $\delta = 1/2$ in the above construction.   Then we can provide the following multiplicative bound, proven in Theorem \ref{thm:powK-up-hat}.  The lower bound holds independent of the choice of $P$ as shown in Theorem \ref{thm:powK-low}.  

\begin{theorem}
\label{thm:powK-bnd}
For any point set $P \subset \Rspace^d$ and point $x \in \Rspace^d$, with empirical measure $\mu_P$ defined by $P$,  then 
\[
\frac{1}{\sqrt{2}} \dK_{\mu_P}(x) \leq \Kpow{\hat P_+}(\mu_P,x) \leq \sqrt{71} \dK_{\mu_P}(x).
\]
\end{theorem}


\subsection{Kernel Power Distance for a Measure $\mu$}
First consider the case for a kernel power distance $\Kpow{P}(\mu,x)$ where $\mu$ is an arbitrary measure.  

\begin{theorem}
\label{thm:powK-low}
For measure $\mu$, point set $P \subset \mathbb{R}^d$, and $x \in \Rspace^d$, 
$
D_K(\mu,x) \leq \sqrt{2} \Kpow{P}(\mu,x).
$
\end{theorem}
\begin{proof}
Let $p = \arg \min_{q \in P}  \left( D_K(q,x)^2 + D_K(\mu,q)^2 \right)$.  Then we can use the triangle inequality and $(D_K(\mu,p) - D_K(p,x))^2 \geq 0$ to show
\[
D_K(\mu,x)^2 
\leq 
(D_K(\mu,p) + D_K(p,x))^2 
\leq 
2 (D_K(\mu,p)^2 + D_K(p,x)^2)
=
2 \Kpow{P}(\mu,x)^2. \qedhere
\]
\end{proof}

\begin{lemma}
\label{lem:pK-up}
For measure $\mu$, point set $P \subset \mathbb{R}^d$, point $p \in P$, and point $x \in \Rspace^d$  then
$
\Kpow{P}(\mu,x)^2 \leq 2 D_K(\mu,x)^2 + 3 D_K(p,x)^2.
$
\end{lemma}
\begin{proof}
Again, we can reach this result with the triangle inequality. 
\begin{align*}
\Kpow{P}(\mu,x)^2 
&\leq 
D_K(\mu,p)^2 + D_K(p,x)^2 
\\ &\leq 
(D_K(\mu,x) + D_K(p,x))^2 + D_K(p,x)^2 
\\ &\leq
2D_K(\mu,x)^2 + 3 D_K(p,x)^2. \qedhere
\end{align*}
\end{proof}

Recall the definition of a point $p_+ = \arg \max_{q \in \Rspace^d} \kappa(\mu,q)$.  
\begin{lemma}
\label{lem:p+}
For any measure $\mu$ and point $x, p_+ \in \Rspace^d$ we have $D_K(p_+,x) \leq 2 D_K(\mu,x)$.  
\end{lemma}
\begin{proof}
Since $x$ is a point in $\Rspace^d$,  $\kappa(\mu,x) \leq \kappa(\mu,p_+)$ and thus $D_K(\mu,x) \geq D_K(\mu,p_+)$.  Then by triangle inequality of $D_K$ to see that 
$D_K(p_+,x) \leq D_K(\mu,x) + D_K(\mu,p_+) \leq 2 D_K(\mu,x)$. 
\end{proof}
 
\begin{theorem}
\label{thm:powK-up}
For any measure $\mu$ in $\Rspace^d$ and any point $x \in \Rspace^d$, using the point $p_+ = \arg \max_{q \in \Rspace^d} \kappa(\mu,q)$ then 
$\Kpow{\{p_+\}}(\mu,x) \leq \sqrt{14} D_K(\mu,x)$.  
\end{theorem}
\begin{proof}
Combine Lemma \ref{lem:pK-up} and Lemma \ref{lem:p+} as
\[
\Kpow{\{p_+\}}(\mu,x)^2 
\leq 
2 D_K(\mu,x)^2 + 3 D_K(p_+,x)^2
\leq
2 D_K(\mu,x)^2 + 3 (4 D_K(\mu,x)^2)
=
14 D_K(\mu,x)^2. \qedhere
\]
\end{proof}

We now need two properties of the point set $P$ to reach our bound, namely, the spread $\beta_P$ and the median concentration $\Lambda_P$. Typically $\log(\beta_P)$ is not too large, and it makes sense to choose $\sigma$ so $\sigma/\Lambda_P \leq 1$, or at least $\sigma/\Lambda_P = O(1)$.


\begin{theorem}
\label{thm:powK-up-hat}
Consider  any point set $P \subset \Rspace^d$ of size $n$, with measure $\mu_P$, spread $\beta_P$, and median concentration $\Lambda_P$.  We can construct a point set $\hat P_+ = P \cup \hat p_+$ in $O(n^2((\sigma/\Lambda_P \delta)^d + \log(\beta))$ time such that for any point $x$, 
$
\Kpow{\hat P_+}(\muP,x) \leq \sqrt{71} D_K(\muP,x).
$  
\end{theorem}
\begin{proof}
We use Theorem \ref{thm:phat+} to find a point $\hpoint_+$ such that $D_K(P,\hat p_+) \leq (3/2) D_K(P,p_+)$. Thus for any $x \in \Rspace^d$, using the triangle inequality
\begin{align*}
D_K(\hat p_+,x) 
&\leq 
D_K(\hat p_+,p_+) + D_K(p_+,x)
\leq
D_K(\muP, \hat p_+) + D_K(\muP,p_+) + D_K(p_+,x)
\\ & \leq 
(5/2) D_K(\muP, p_+) + D_K(p_+,x).
\end{align*}
Now combine this with Lemma \ref{lem:pK-up} and Lemma \ref{lem:p+} as
\begin{align*}
\Kpow{\hat P_+}(\muP,x)^2 
&\leq 
2 D_K(\muP,x)^2 + 3 D_K(\hat p_+,x)^2
\\ &\leq
2 D_K(\muP,x)^2 + 3 ((5/2) D_K(\muP,x) + D_K(p_+,x))^2
\\ &\leq
2 D_K(\muP,x)^2 + 3 ((25/4) + (5/2)) D_K(\muP,x)^2 + (1+5/2) D_K(p_+,x)^2)
\\ & = 
(113/4) D_K(\muP,x)^2 + (21/2) D_K(p_+,x)^2
\\ &\leq
(113/4) D_K(\muP,x)^2 + (21/2) (4 D_K(\muP,x)^2)
<  
71 D_K(\muP,x)^2.  \qedhere
\end{align*}
\end{proof}

\section{Reconstruction and Topological Estimation using Kernel Distance}
\label{sec:recon}

Now applying distance-like properties from Section \ref{sec:prop} and the power distance properties of Section \ref{sec:powerdistance} we can apply known reconstruction results to the kernel distance.  

\subsection{Homotopy Equivalent Reconstruction using $d^K_{\mu}$}

We have shown that the kernel distance function $d^K_{\mu}$ is a distance-like function.
Therefore the reconstruction theory for a distance-like function \cite{ChazalCohen-SteinerMerigot2011} (which is an extension of results for compact sets \cite{ChazalCohen-SteinerLieutier2009}) 
holds in the setting of $d^K_{\mu}$. 
We state the following two corollaries for completeness, whose proofs follow from the proofs of Proposition 4.2 and Theorem 4.6 in \cite{ChazalCohen-SteinerMerigot2011}.  
Before their formal statement, we need some notation adapted from \cite{ChazalCohen-SteinerMerigot2011} to make these statements precise.
Let $\phi: \Rspace^d \to \Rspace^+$ be a distance-like function. A point $x \in \Rspace^d$ is an $\alpha$-critical point if $\phi^2(x+h) \leq \phi^2(x) + 2\alpha \| h\| \phi(x) + \| h\|^2$ with $\alpha \in [0,1]$, $\forall h \in \Rspace^d$. 
Let $(\phi)^r = \{x \in \Rspace^d \mid \phi(x) \leq r\}$ denote the sublevel set of $\phi$, and let $(\phi)^{[r_1, r_2]} = \{x \in \Rspace^d \mid r_1 \leq \phi(x) \leq r_2\}$ denote all points at levels in the range $[r_1, r_2]$.  
For $\alpha \in [0,1]$, the \emph{$\alpha$-reach} of $\phi$ is the maximum $r$ such that $(\phi)^r$ has no $\alpha$-critical point, denoted as $\reach_{\alpha}(\phi)$.   
When $\alpha=1$, $\reach_1$ coincides with \emph{reach} introduced in \cite{Federer1959}.

\begin{theorem}[Isotopy lemma on $d^K_{\mu}$]
Let $r_1 < r_2$ be two positive numbers such that $d^K_{\mu}$ has no critical points in $(\dK_\mu)^{[r_1,r_2]}$.  
Then all the sublevel sets $(\dK_\mu)^r$ are isotopic for $r \in [r_1, r_2]$. 
\end{theorem}

\begin{theorem}[Reconstruction on $\dK_\mu$]
Let $\dK_\mu$ and $\dK_\nu$ be two kernel distance functions such that
$\|\dK_\mu - \dK_\nu\|_\infty \leq \eps$.  
Suppose $\reach_{\alpha}(\dK_\mu) \geq R$ for some $\alpha > 0$.
Then $\forall r \in [4\eps/\alpha^2, R - 3\eps]$, and $\forall \eta \in (0, R)$, 
the sublevel sets $(\dK_\mu)^\eta$ and $(\dK_\nu)^r$ are homotopy equivalent for $\eps \leq R/(5+4/\alpha^2)$. 
\label{thm:reconstruct}
\end{theorem}

\subsection{Constructing Topological Estimates using $d^K_{\mu}$}
\label{subsec:topo-est}
In order to actually construct a topological estimate using the kernel distance $d^K_{\mu}$, one needs to be able to compute quantities related to its sublevel sets, in particular, to compute the persistence diagram of the sub-level sets filtration of $d^K_{\mu}$.  
Now we describe such tools needed for the kernel distance based on machinery recently developed by Buchet et al.  \cite{BuchetChazalOudot2013}, which shows how to approximate the persistent homology of distance-to-a-measure for any metric space via a power distance construction. Then using similar constructions, we can use the weighted Rips filtration to approximate the persistence diagram of the kernel distance. 

To state our results, first we require some technical notions and assume basic knowledge on persistent homology (see \cite{EdelsbrunnerHarer2008,EdelsbrunnerHarer2010} for a readable background). 
Given a metric space $\Xspace$ with the distance $d_{\Xspace}(\cdot, \cdot)$, a set $P \subseteq  \Xspace$ and a function $w: P \to \Rspace$, the (general) \emph{power distance} $f$ associated with $(P, w)$ is defined as 
$
f(x) = \sqrt{\min_{p \in P} \left( d_{\Xspace}(p,x)^2 + w(p)^2  \right)}.
$ 
Now given the set $(P,w)$ and its corresponding power distance $f$, one could use the weighted Rips filtration to approximate the persistence diagram of $w$, under certain restrictive conditions proven in Appendix \ref{subsec:approximation}.  
Consider the sublevel set of $f$, $f^{-1}((-\infty, \alpha])$. It is the union of balls centered at points $p \in P$ with radius $r_p(\alpha) = \sqrt{\alpha^2 - w(p)^2}$ for each $p$.
The weighted \v{C}ech complex $C_{\alpha}(P, w)$ for parameter $\alpha$ is the union of simplices $s$ such that $\bigcap_{p \in s} B(p, r_p(\alpha)) \neq 0$. 
The weighted Rips complex $R_{\alpha}(P, w)$ for parameter $\alpha$ is the maximal complex whose $1$-skeleton is the same as $C_{\alpha}(P, w)$.  
The corresponding weighted Rips filtration is denoted as $\{R_{\alpha}(P, w)\}$. 

Setting $w: = \dK_{\mu_P}$ and given point set $\hat{P}_+$ described in Section \ref{sec:powerdistance},  consider the weighted Rips filtration $\{R_{\alpha}(\hat{P}_+,\dK_\mu)\}$ based on the kernel power distance, $\Kpow{\hat{P}_+}$. 
We view the persistence diagrams on a logarithmic scale, that is, we change coordinates of points following the mapping $(x, y) \mapsto (\ln x, \ln y)$. $d_B^{\ln}$ denotes the corresponding bottleneck distance between persistence diagrams.  We now state a corollary of Theorem \ref{thm:powK-bnd}.
\begin{corollary}
\label{cor:wRip}
The weighted Rips filtration $\{R_{\alpha}(\hP_+,\dK_\muP)\}$ can be used to  approximate the persistence diagram of $\dK_{\mu_P}$ such that 
$d_B^{\rm \ln}(\dgmD{\dK_{\mu_P}}, \dgmD{\{R_{\alpha}(\hP_+,\dK_{\mu_P})\})} \leq \ln(2\sqrt{71})$.
\end{corollary}
\begin{proof}
To prove that two persistence diagrams are close, one could prove that their filtration are interleaved \cite{ChazalCohen-SteinerGlisse2009}, 
that is, two filtrations $\{U_{\alpha}\}$ and $\{V_{\alpha}\}$ are \emph{$\eps$-interleaved} if for any $\alpha$, $U_{\alpha} \subseteq V_{\alpha+\eps} \subseteq U_{\alpha+2\eps}$.

First, Lemmas \ref{lemma:q-tame} and \ref{lemma:q-tame-rips} prove that the persistence diagrams $\dgmD{\dK_{\mu_P}}$ and $\dgmD{\{R_{\alpha}(\hP_+,\dK_{\mu_P})\}})$ are well-defined. 
Second, the results of Theorem \ref{thm:powK-bnd} implies an $\sqrt{71}$ multiplicative interleaving. 
Therefore for any $\alpha \in \Rspace$, 
\[
{(\dK_\muP)}^{-1}((-\infty, \alpha]) \subset {(\Kpow{\hat P_+})}^{-1}((-\infty, \sqrt{2}\alpha) \subset {(\dK_{\muP})}^{-1}((-\infty, \sqrt{71}\sqrt{2}\alpha]).
\]
On a logarithmic scale (by taking the natural log of both sides), such interleaving becomes addictive, 
\[
\ln \dK_\muP - \sqrt{2} \leq \ln \Kpow{\hat P_+} \leq \ln \dK_\muP + \sqrt{71}.  
\]
Theorem 4 of \cite{ChazalSilvaGlisse2013} implies 
\[
d_B^{\rm \ln}(\dgmD{\dK_\muP}, \dgmD{\Kpow{\hat P_+}}) \leq \sqrt{71}.
\]
In addition, by the Persistent Nerve Lemma (\cite{ChazalOudot2008}, Theorem 6 of \cite{Sheehy2012}, an extension of the Nerve Theorem \cite{Hatcher2002}), 
the sublevel sets filtration of $\dK_\mu$, which correspond to unions of balls of increasing radius, has the same persistent homology as the nerve filtration of these balls (which, by definition, is the \v{C}ech filtration).
Finally, there exists a multiplicative interleaving between weighted Rips and \v{C}ech complexes (Proposition 31 of \cite{ChazalSilvaGlisse2013}), 
$
C_{\alpha} \subseteq R_{\alpha} \subseteq C_{2\alpha}. 
$
We then obtain the following bounds on persistence diagrams,  
\[
d_B^{\rm \ln}(\dgmD{\Kpow{P_+}}, \dgmD{\{R_{\alpha}(P_+,\dK_\muP)\}}) \leq \ln(2).
\]
We use triangle inequality to obtain the final result: 
\[
d_B^{\rm \ln}(\dgmD{\dK_\muP}, \dgmD{\{R_{\alpha}(P_+,\dK_\muP)\}}) \leq \ln(2\sqrt{71}). \qedhere
\]
\end{proof}

Based on Corollary \ref{cor:wRip}, we have an algorithm that approximates the persistent homology of the sublevel set filtration of $\dK_\mu$ by constructing the weighted Rips filtration corresponding to the kernel-based power distance and computing its persistent homology. 
For memory efficient computation, sparse (weighted) Rips filtrations could be adapted by considering simplices on subsamples at each scale \cite{Sheehy2013,ChazalSilvaGlisse2013}, although some restrictions on the space apply.

\subsection{Distance to the Support of a Measure vs. Kernel Distance}
\label{sec:infer}
Suppose $\mu$ is a uniform measure on a compact set $S$ in $\Rspace^d$. 
We now compare the kernel distance $\dK_\mu$ with the distance function $f_S$ to the support $S$ of $\mu$. We show how $\dK_\mu$ approximates $f_S$, and thus allows one to infer geometric properties of $S$ from samples from $\mu$. 

A generalized gradient and its corresponding flow associated with a distance function are described in \cite{ChazalCohen-SteinerLieutier2009} and later adapted for distance-like functions in \cite{ChazalCohen-SteinerMerigot2011}. 
Let $f_S: \Rspace^d \to \Rspace$ be a distance function associated with a compact set $S$ of $\Rspace^d$. 
It is not differentiable on the medial axis of $S$. 
A \emph{generalized gradient function} $\grad{S}{}: \Rspace^d \to \Rspace^d$ coincides with the usual gradient of $f_S$ where $f_S$ is differentiable, and is defined everywhere and can be integrated into a continuous flow $\Phi^t: \Rspace^d \to \Rspace^d$ that points away from $S$.   
Let $\gamma$ be an integral (flow) line.  
The following lemma shows that when close enough to $S$, that $\dK_\mu$ is strictly increasing along any $\gamma$.  The proof is quite technical and is thus deferred to Appendix \ref{app:infer}.  

\begin{lemma}
Given any flow line $\gamma$ associated with the generalized gradient function $\grad{S}{}$, $\dK_\mu(x)$ is strictly monotonically increasing along $\gamma$ for $x$ sufficiently far away from the medial axis of $S$, for $\sigma \leq \frac{R}{6\Delta_G}$ and $ f_S(x) \in (0.014 R, 2\sigma)$. 
Here $B(\sigma/2)$ denotes a ball of radius $\sigma/2$,  $G := \frac{\vol(B(\sigma/2))}{\vol(S)}$, $\Delta_G : = \sqrt{12 + 3\ln(4/G)}$ and suppose $R := \min(\reach(S), \reach(\Rspace^d \setminus S)) > 0$.  
\label{lem:monotonicity}
\end{lemma}

The strict monotonicity of $\dK_\mu$ along the flow line under the conditions in Lemma \ref{lem:monotonicity} makes it possible to define a deformation retract of the sublevel sets of $\dK_\mu$ onto sublevel sets of $f_S$. Such a deformation retract defines a special case of homotopy equivalence between the sublevel sets of $\dK_\mu$ and sublevel sets of $f_S$. 
Consider a sufficiently large point set $P \in \mathbb{R}^d$ sampled from $\mu$, and its induced measure $\muP$.  
We can then also invoke Theorem \ref{thm:reconstruct} and a sampling bound (see Section \ref{sec:algo} and Lemma \ref{lem:KD-samp-nosq}) to show homotopy equivalence between the sublevel sets of $f_S$ and $\dK_\muP$.  


Note that Lemma \ref{lem:monotonicity} uses somewhat restrictive conditions related to the reach of a compact set, however we believe such conditions could be further relaxed to be associated with the concept of $\mu$-reach as described in \cite{ChazalCohen-SteinerLieutier2009}.

\section{Stability Properties for the Kernel Distance to a Measure}

\begin{lemma}[K4]
For two measures $\mu$ and $\nu$ on $\Rspace^d$ we have
$\| \dK_\mu - \dK_\nu \|_\infty \leq D_K(\mu,\nu)$. 
\end{lemma}
\begin{proof}
Since $D_K(\cdot, \cdot)$ is a metric, then by triangle inequality,  for any $x \in \Rspace^d$ we have 
$D_K(\mu,x) \leq D_K(\mu,\nu) + D_K(\nu,x)$ and 
$D_K(\nu,x) \leq D_K(\nu,\mu) + D_K(\mu, x)$. 
Therefore for any $x \in \Rspace^d$ we have $| D_K(\mu,x) - D_K(\nu,x) | \leq D_K(\mu,\nu)$, proving the claim.  
\end{proof}

Both the Wasserstein and kernel distance are \emph{integral probability metrics}~\cite{SGFSL10}, so (M4) and (K4) are both interesting, but not easily comparable.  We now attempt to reconcile this.  

\subsection{Comparing $D_K$ to $W_2$}
\begin{lemma}
There is no Lipschitz constant $\gamma$ such that for any two probability measures $\mu$ and $\nu$ we have $W_2(\mu,\nu) \leq \gamma D_K(\mu,\nu)$.  
\end{lemma}
\begin{proof}
Consider two measures $\mu$ and $\nu$ which are almost identical: the only difference is some mass of measure $\tau$ is moved from its location in $\mu$ a distance $n$ in $\nu$. 
The Wasserstein distance requires a transportation plan that moves this $\tau$ mass in $\nu$ back to where it was in $\mu$ with cost $\tau \cdot \Omega(n)$ in $W_2(\mu,\nu)$.  
On the other hand,  $D_K(\mu,\nu) = \sqrt{\kappa(\mu,\mu) + \kappa(\nu,\nu) - 2 \kappa(\mu,\nu)} \leq \sqrt{\sigma^2 + \sigma^2 - 2\cdot0} = \sqrt{2} \sigma$ is bounded.  
\end{proof}

We conjecture for any two probability measures $\mu$ and $\nu$ that $D_K(\mu,\nu) \leq W_2(\mu,\nu)$.  This would show that $\dK_\mu$ is at least as stable as $\dC_{\mu,m_0}$ since a bound on $W_2(\mu,\nu)$ would also bound $D_K(\mu,\nu)$, but not vice versa.  
Alternatively, this can be viewed as $\dK_\mu$ is less discriminative than $\dC_{\mu,m_0}$; we view this as a positive in this setting, as it is mainly less discriminative towards outliers (far away points).  
Here we only show that this property for a special case and as $\sigma \to \infty$.  
To simplify notation, all integrals are assumed to be over the full domain $\Rspace^d$.  

\paragraph{Two Dirac masses.}
We first consider a special case when $\mu$ is a Dirac mass at a point $p$ and $\nu$ is a Dirac mass at a point $q$.  That is they are both single points.  We can then write $D_K(\mu,\nu) = D_K(p,q)$.  Figure \ref{fig:Euc-bound} illustrates the result of this lemma.  

\begin{lemma}
\label{lem:E2DK}
For any points $p,q \in \Rspace^d$ it always holds that $\|p-q\| \geq D_K(p,q)$.  
When $\|p-q\| \leq \sqrt{3} \sigma$ then $D_K(p,q) \geq \|p-q\|/2$.  
\end{lemma}
\begin{proof}
First expand $D_K(p,q)^2$ as
\[
D_K(p,q)^2 
= 
2 \sigma^2 - 2 K(p,q)
= 
2 \sigma^2 \left(1 - \exp\left(\frac{-\|p-q\|^2}{2\sigma^2}\right)\right).
\]
Now using that $1 - t \leq e^{-t}  \leq 1 - t + (1/2) t^2$ for $t\geq0$
\[
D_K(p,q)^2 = 2 \sigma^2 \left(1 - \exp\left(\frac{-\|p-q\|^2}{2\sigma^2}\right)\right)
\leq
2 \sigma^2 \left(\frac{\|p-q\|^2}{2\sigma^2}\right)
= 
\|p-q\|^2
\]
and 
\begin{align*}
D_K(p,q)^2 &= 2 \sigma^2 \left(1 - \exp\left(\frac{-\|p-q\|^2}{2\sigma^2}\right)\right)
\\ &\geq
2 \sigma^2 \left(\frac{\|p-q\|^2}{2\sigma^2} - \frac{1}{2} \frac{\|p-q\|^4}{4 \sigma^4}\right)
\\ &=
\frac{\|p-q\|^2}{4} \left(4 - \frac{\|p-q\|^2}{\sigma^2}\right)
\\ &\geq
\|p-q\|^2/4,
\end{align*}
where the last inequality holds when $\|p-q\|^2 \leq  \sqrt{3}\sigma$.  
\end{proof}

\begin{figure}
\begin{center}\includegraphics[width=.9\linewidth]{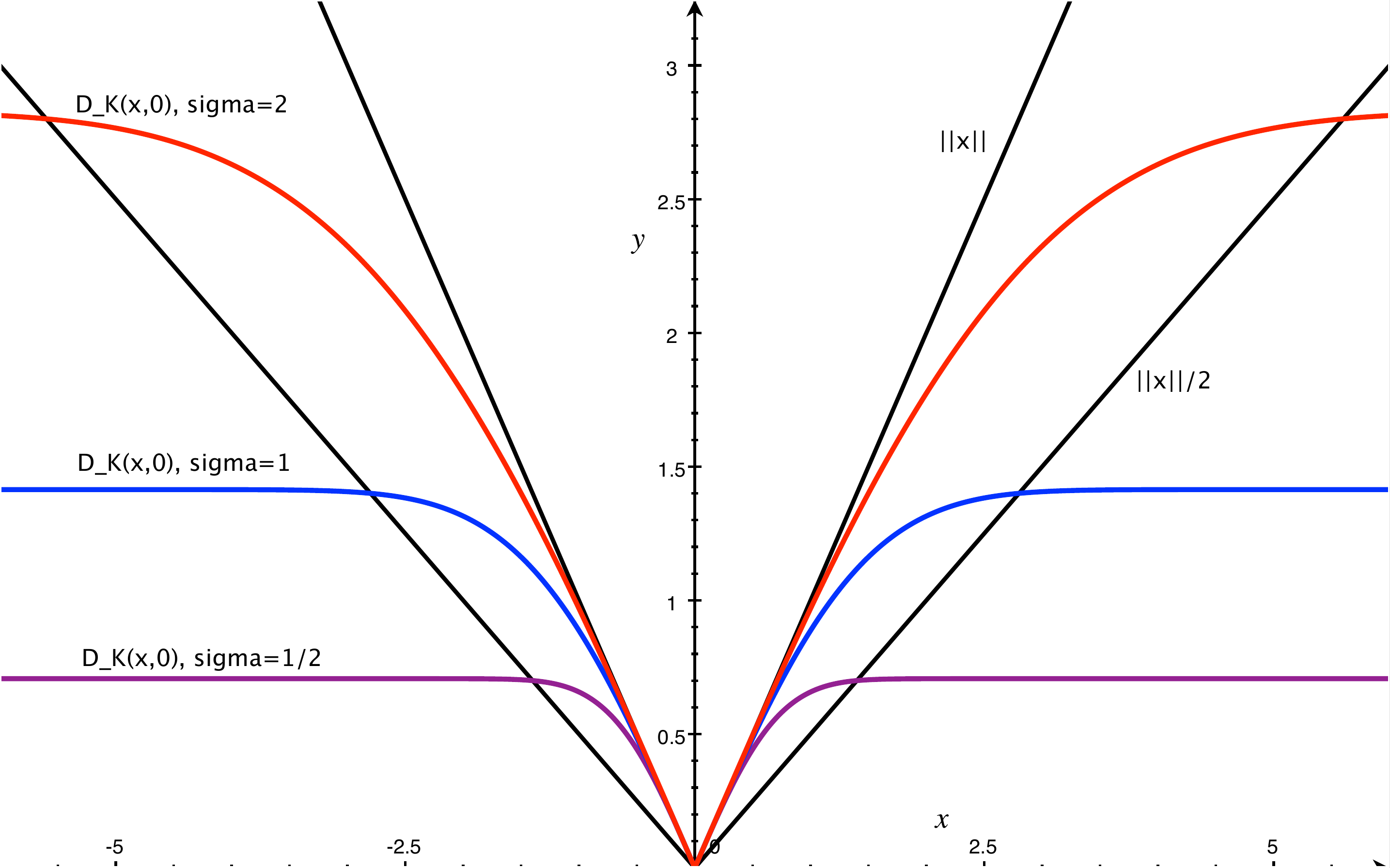}\end{center}
\caption{\label{fig:Euc-bound}
\small \sffamily
Showing that $\|x-0\|/2 \leq D_K(x,0) \leq \|x-0\|$, where the second inequality holds for $\|x\| \leq \sqrt{3} \sigma$.  The kernel distance $D_K(x,0)$ is shown for $\sigma = \{1/2, 1, 2\}$ in purple, blue, and red, respectively.  
}
\end{figure}

\paragraph{One Dirac mass.}
Consider the case where one measure $\nu$ is a Dirac mass at point $x \in \Rspace^d$.    

\begin{lemma}
Consider two probability measures $\mu$ and $\nu$ on $\Rspace^d$ where $\nu$ is represented by a Dirac mass at a point $x \in \Rspace^d$.  
Then $\dK_\mu(x) = D_K(\mu,\nu) \leq W_2(\mu,\nu)$ for any $\sigma >0$, where the equality only holds when $\mu$ is also a Dirac mass at $x$.
\label{lem:2pts}
\end{lemma}
\begin{proof}
Since both $W_2(\mu,\nu)$ and $D_K(\mu,\nu)$ are metrics and hence non-negative, we can instead consider their squared versions: $(W_2(\mu,\nu))^2 = \int_p \|p-x\|^2 \mu(p) \dir p$ and 
\vspace{-2mm}
\begin{align*}
(D_K(\mu,\nu))^2 
&= 
K(x,x) + \int_{(p,q)} K(p,q) \dir{\textsf{m}_{\mu,\mu}(p,q)} - 2 \int_{p} K(p,x) \dir{\mu(p)}
\\ &= 
\sigma^2 \left(1 + \int_{(p,q)} \exp \left(-\frac{\|p-q\|^2}{2 \sigma^2} \right) \dir{\textsf{m}_{\mu,\mu}(p,q)} - 2 \int_p \exp \left( - \frac{\|p-x\|^2}{2 \sigma^2} \right) \dir{\mu(p)} \right).
\end{align*}
Now use the bound $1-t \leq e^{-t} \leq 1$ for $t \geq 0$ to approximate 
\begin{align*}
(D_K(\mu,\nu))^2 
& \leq
\sigma^2 \left(1 + \int_{(p,q)} (1) \dir{(\textsf{m}_{\mu,\mu}(p,q)} - 2 \int_p \left(1-\frac{\|p-x\|^2}{2 \sigma^2} \right) \dir{\mu(p)} \right)
\\ & =
\int_p \|p-x\|^2 \dir{\mu(p)} = (W_2(\mu,\nu))^2. 
\end{align*}
The inequality becomes an equality only when $\|p-x\| = 0$ for all $p \in P$, and since they are both metrics, this is the only location where they are both $0$.  
\end{proof}

\paragraph{General case.}
Next we show that if $\nu$ is not a unit Dirac, then this inequality holds in the limit as $\sigma$ goes to infinity.  The technical work is making precise how $\sigma^2 - K(p,x) \leq \|x-p\|^2/2$ and how this compares to bounds on $D_K(\mu,\nu)$ and $W_2(\mu,\nu)$.  

For simpler exposition, we assume $\mu$ is a probability measure, that is $\int_p \mu(p) \dir p = 1$; otherwise we can normalize $\mu$ at the appropriate locations, and all of the results go through.  

\begin{lemma}
For any $p,q \in \Rspace^d$ we have
$\begin{displaystyle}
K(p,q) 
=
\sigma^2 - \frac{\|p-q\|^2}{2} + \sum_{i=2}^\infty \frac{(-\|p-q\|^2)^i}{2^{i+1} \sigma^{2i-2} i!}. 
\end{displaystyle}$
\label{lem:Kflip}
\end{lemma}
\begin{proof}
We use the Taylor expansion of $e^x = \sum_{i=0}^\infty x^i/i! = 1 + x + \sum_{i=2}^\infty x^i /i!$.
Then it is easy to see
\[
K(p,q) 
= 
\sigma^2 \exp \left(-\frac{\|p-q\|^2}{2 \sigma^2} \right) 
=
\sigma^2 - \frac{\|p-q\|^2}{2} + \sum_{i=2}^\infty \frac{(-\|p-q\|^2)^i}{2^i \sigma^{2i-2} i!}. \qedhere
\]
\end{proof}

This lemma illustrates why the choice of coefficient of $\sigma^2$ is convenient.  Since then $\sigma^2 - K(p,q)$ acts like $\frac{1}{2}\|p-q\|^2$, and becomes closer as $\sigma$ increases.  
Define $\bar \mu = \int_p p \cdot  \dir \mu(p)$ to represent the mean point of measure $\mu$;
$\Var(\mu) = (\int_p \|p\|^2 \dir{\mu(p)}) - \|\bar \mu\|^2$  to represent the variance of the measure $\mu$; and  
$\Delta_{\mu,\nu} = \int_{(p,q)} \sum_{i=2}^\infty \frac{(-\|p-q\|^2)^i}{2^i \sigma^{2i-2} i!} \dir{\textsf{m}_{\mu,\nu}(p,q)}$.

\begin{lemma}
For any $x \in \Rspace^d$ we have
$\begin{displaystyle}
\int_p \|p - x\|^2 \dir{\mu(p)} = \|\bar \mu - x\|^2 + \Var(\mu).
\end{displaystyle}$
\label{lem:mean+Var}
\end{lemma}
\begin{proof} 
\begin{align*}
\int_p \|p - x\|^2 \dir{\mu(p)}
&= 
\int_p \left( \|p\|^2 + \|x\|^2 - 2 \langle p , x \rangle \right) \dir{\mu(p)}
\\ &=
\int_p \|p\|^2 \dir{\mu(p)} + \|x\|^2 - 2 \int_p \langle p , x \rangle \dir{\mu(p)}
\\ &=
\left(\int_p \|p\|^2 \dir{\mu(p)} - \|\bar \mu\|^2\right) + \|\bar \mu\|^2 + \|x\|^2 - 2 \langle \bar \mu , x \rangle
\\ &=
\Var(\mu) + \|\bar \mu - x\|^2.  \qedhere
\end{align*}
\end{proof}

\begin{lemma} 
For probability measures $\mu$ and $\nu$ on $\Rspace^d$,  
$
\kappa(\mu,\nu) 
=
\sigma^2 - \frac{1}{2} \left(\|\bar \mu - \bar \nu\|^2 + \Var(\mu) + \Var(\nu) \right) + \Delta_{\mu,\nu}.
$
\label{lem:decomp-kap}
\end{lemma}
\begin{proof}
We use Lemma \ref{lem:Kflip} to expand
\begin{align*}
\kappa(\mu,\nu) 
&= 
\int_{(p,q)} K(p,q) \dir{\textsf{m}_{\mu,\nu}(p,q)}
\\ &=
\sigma^2 - \int_{(p,q)}\left( \frac{\|p-q\|^2}{2} - \sum_{i=2}^\infty \frac{(-\|p-q\|^2)^i}{2^{i+1} \sigma^{2i-2} i!}\right)  \dir{\textsf{m}_{\mu,\nu}(p,q)}.
\end{align*}
After shifting the $\Delta_{\mu,\nu}$ term outside, we can use Lemma \ref{lem:mean+Var} (twice) to rewrite
\begin{align*}
\int_p \left( \int_q \|p-q\|^2 \dir{\nu(q)} \right) \dir{\mu(p)}
&=
\int_p \left(\|p - \bar \nu\|^2 + \Var(\nu)\right) \dir{\mu(p)}
\\ &=
\|\bar \mu - \bar \nu\|^2 + \Var(\mu) + \Var(\nu). \qedhere
\end{align*}
\end{proof}

\begin{theorem}
For any two probability measures $\mu$ and $\nu$ defined on $\Rspace^d$ 
$\begin{displaystyle}
\lim_{\sigma \to \infty} D_K(\mu,\nu) = \|\bar \mu - \bar \nu\|.
\end{displaystyle}$  
\label{thm:DK-means}
\end{theorem}
\begin{proof} 
First expand 
\begin{align*}
(D_K(\mu,\nu))^2 
= & \;
\kappa(\mu,\mu) + \kappa(\nu,\nu) - 2 \kappa(\mu,\nu)
\\ = & 
\left(\sigma^2 - \frac{1}{2} \|\bar \mu - \bar \mu\|^2 - \Var(\mu)  + \Delta_{\mu,\mu} \right)
 + 
\left(\sigma^2 - \frac{1}{2} \|\bar \nu - \bar \nu\|^2 - \Var(\nu)  + \Delta_{\nu,\nu} \right)
 \\ & - 
2\left(\sigma^2 - \frac{1}{2} \|\bar \mu - \bar \nu\|^2 - \frac{1}{2} \Var(\mu) - \frac{1}{2}\Var(\nu)  + \Delta_{\mu,\nu} \right)
\\ =& \;
\|\bar \mu - \bar \nu\|^2 + \Delta_{\mu,\mu} + \Delta_{\nu,\nu} - 2 \Delta_{\mu,\nu}.
\end{align*}
Finally we observe that since all terms of $\Delta_{\mu,\nu}$ are divided by $\sigma^2$ or larger powers of $\sigma$.  Thus as $\sigma$ increases $\Delta_{\mu,\nu}$ approaches $0$ and $(D_K(\mu,\nu))^2$ approaches $\|\bar \mu - \bar \nu\|^2$, completing the proof.  
\end{proof}

Now we can relate $D_K(\mu,\nu)$ to $W_2(\mu,\nu)$ through $\|\bar \mu - \bar \nu\|$.  
The next result is a known lower bounds for the Earth movers distance~\cite{Coh99}[Theorem 7].  We reprove it in Appendix \ref{app:stability} for completeness.

\begin{lemma} 
For any probability measures $\mu$ and $\nu$ defined on $\Rspace^d$ we have
$\|\bar \mu - \bar \nu\| \leq W_2(\mu,\nu).$
\label{lem:EMD}
\end{lemma}

We can now combine these results to achieve the following theorem.  

\begin{theorem}
\label{thm:DK-W2}
For any two probability measures $\mu$ and $\nu$ defined on $\Rspace^d$ 
$\begin{displaystyle}
\lim_{\sigma \to \infty} D_K(\mu,\nu) = \|\bar \mu - \bar \nu\|
\end{displaystyle}$  and $\|\bar \mu - \bar \nu\| \leq W_2(\mu,\nu).$  Thus
$\lim_{\sigma \to \infty} D_K(\mu,\nu) \leq W_2(\mu,\nu)$.  
\end{theorem}

\subsection{Kernel Distance Stability with Respect to $\sigma$}
\label{subsection:sigma}
We now explore the Lipschitz properties of $\dK_\mu$ with respect to the noise parameter $\sigma$.  
We argue any distance function that is robust to noise needs some parameter to address how many outliers to ignore or how far away a point is that is an  outlier.  For instance, this parameter in $\dC_{\mu, m_0}$ is $m_0$ which controls the amount of measure $\mu$ to be used in the distance.  

Here we show that $\dK_\mu$ has a particularly nice property, that it is Lipschitz with respect to the choice of $\sigma$ for any fixed $x$.  
The larger $\sigma$ the more effect outliers have, and the smaller $\sigma$ the less the data is smoothed and thus the closer the noise needs to be to the underlying object to effect the inference.  

\begin{lemma}
Let $h(\sigma,z) = \exp(-z^2/2\sigma^2)$.  We can bound
$h(\sigma,z) \leq 1$, 
$\frac{\dir{}}{\dir \sigma} h(\sigma,z) \leq (2/e)/\sigma$ and 
$\frac{\dir{}^2}{\dir \sigma^2} h(\sigma,z) \leq (18/e^{3})/ \sigma^2$ over any choice of $z>0$.  
\label{lem:h-max}
\end{lemma}
\begin{proof}
The first bound follows from $y = -z^2/2\sigma^2 \leq 0$ and $\exp(y) \leq 1$ for $y\leq 0$. 

Next we define
\begin{align*}
w_1(\sigma,z) &= \frac{\dir{}}{\dir \sigma} h(\sigma,z) = \frac{z^2}{\sigma^3} \exp\left(\frac{-z^2}{2 \sigma^2}\right), \text{ and}
\\
w_2(\sigma,z) &= \frac{\dir{}^2}{\dir \sigma^2} h(\sigma,z) = \left(\frac{z^4}{\sigma^6} - \frac{3 z^2}{\sigma^4} \right) \exp\left(\frac{-z^2}{2 \sigma^2}\right). 
\end{align*}

Now to solve the first part, we differentiate $w_1$ with respect to $z$ to find its maximum over all choices of $z$.  
\[
\frac{\dir{}}{\dir z} w_1(\sigma,z) = \left(\frac{2 z}{\sigma^3} - \frac{z^3}{\sigma^5}\right) \exp\left(\frac{-z^2}{2 \sigma^2}\right)
\]
Where $(\dir{} / \dir z) w_1(\sigma,z) = 0$ at $z=0$, $z=\sqrt{2}\sigma$ and as $z$ approaches $\infty$.  Thus the maximum must occur at one of these values.  Both $w_1(\sigma,0) = 0$ and $\lim_{z \to \infty} w_1(\sigma,z) = 0$, while $w_1(\sigma,\sqrt{2}\sigma) = (2/e)/\sigma$, proving the first part.  

To solve the second part, we perform the same approach on $w_2$.  
\begin{align*}
\frac{\dir{}}{\dir z} w_2(\sigma,z) 
&= 
\left(\frac{- z^5}{\sigma^8} + \frac{3 z^3}{\sigma^6} + \frac{4z^3}{\sigma^6} - \frac{6z}{\sigma^4}\right) \exp\left(\frac{-z^2}{2 \sigma^2}\right)
\\ &= 
\frac{z}{\sigma^4} \left(\frac{- z^4}{\sigma^4} + \frac{7 z^2}{\sigma^2} - 6\right) \exp\left(\frac{-z^2}{2 \sigma^2}\right) 
\end{align*}
Thus $(\dir{} / \dir z) w_2(\sigma,z) = 0$ at $z=\{0, \sigma, \sqrt{6}\sigma\}$ and as $z$ goes to $\infty$ for $z \in [0,\infty)$.  Both $w_2(\sigma,0) = 0$ and $\lim_{z \to \infty} w_2(\sigma,z) = 0$.  The minimum occurs at $w_2(\sigma,z=\sigma) = (-2/\sqrt{e})/\sigma^2$.  
The maximum occurs at $w_2(\sigma,z=\sqrt{6}\sigma) = (18/e^{3})/\sigma^2$.  
\end{proof}

\begin{theorem}
\label{thm:sigma-Lip}
For any measure $\mu$ defined on $\mathbb{R}^d$ and $x \in \mathbb{R}^d$, $\dK_\mu(x)$ is $\ell$-Lipschitz with respect to $\sigma$, for $\ell = 18/e^3 + 8/e + 2 < 6$.
\end{theorem}
\begin{proof}
Recall that $\textsf{m}_{\mu,\nu}$ is the product measure of any $\mu$ and $\nu$, and
define $\textsf{M}_{\mu,\nu}$ as $\textsf{M}_{\mu,\nu}(p,q) = \textsf{m}_{\mu,\mu}(p,q) + \textsf{m}_{\nu,\nu}(p,q) - 2 \textsf{m}_{\mu,\nu}(p,q)$.  
It is now useful to define a function $f_x(\sigma)$ as  
\[
f_x(\sigma) = \int_{(p,q)} \exp\left(\frac{-\|p-q\|^2}{2\sigma^2}\right) \dir{\textsf{M}_{\mu,\delta_x}(p,q)}. 
\]
So $\dK_\mu(x) = \sigma \sqrt{f_x(\sigma)}$ and we can write another function as
\[
F(\sigma) = (\dK_\mu(x))^2 - \ell \|\sigma\|^2 = \sigma^2 f_x(\sigma) - \ell \sigma^2.  
\]
Now to prove  $\dK_\mu(x)$ is $\ell$-Lipschitz, we can show that $(\dK_\mu)^2$ is $\ell$-semiconcave with respect to $\sigma$, and apply Lemma \ref{lem:semi-Lip}.  This boils down to showing the second derivative of $F(\sigma)$ is always non-positive.  

\[
\frac{\dir{}}{\dir \sigma} F(\sigma) 
= 
2 \sigma f_x(\sigma)  + \sigma^2 \frac{\dir{}}{\dir \sigma} f_x(\sigma)  - 2 \sigma \ell.
\]

\[
\frac{\dir{}^2}{\dir \sigma^2} F(\sigma) 
= 
 \sigma^2 \frac{\dir{}^2}{\dir \sigma^2} f_x(\sigma)  +  4 \sigma \frac{\dir{}}{\dir \sigma} f_x(\sigma)  + 2 f_x(\sigma) - 2 \ell.
\]
First we note that since $\int_{(p,q)} c \cdot \dir{} \textsf{m}_{\mu,\nu}(p,q) = c$ for any product distribution $\textsf{m}_{\mu,\nu}$ of two distributions $\mu$ and $\nu$, including when $\mu$ or $\nu$ is a Dirac mass, then 
\[
\int_{(p,q)} c \cdot \dir{}\textsf{M}_{\mu,\delta_x}(p,q) = 
\int_{(p,q)} c \cdot \dir{}\Big[\textsf{m}_{\mu,\mu} + \textsf{m}_{\delta_x,\delta_x} - 2 \textsf{m}_{\mu,\delta_x}\Big] (p,q) \leq 2c.
\]
Thus since $\exp \left( \frac{-\|p-q\|^2}{2 \sigma^2} \right)$ is in $[0,1]$ for all choices of $p$, $q$, and $\sigma >0$, then $0 \leq f_x(\sigma) \leq 2$ and $2 f_x(\sigma) \leq 4$.  This bounds the third term in $\frac{\dir{}^2}{\dir \sigma^2} F(\sigma)$, we now need to use a similar approach to bound the first and second terms.  

Let $h(\sigma,z) = \exp\left(\frac{-z^2}{2 \sigma^2}\right)$, so we can apply Lemma \ref{lem:h-max}.
\[
4 \sigma \frac{\dir{}}{\dir \sigma} f_x(\sigma) 
=
 4 \sigma \int_{(p,q)} \left(\frac{\dir{}}{\dir \sigma} h(\sigma,\|p-q\|)\right) \dir{}\textsf{M}_{\mu,\delta_x}(p,q)
\leq  
4\sigma ((2/e)/\sigma) 2 = 16/e
\]

\[
\sigma^2 \frac{\dir{}^2}{\dir \sigma^2} f_x(\sigma) 
= 
\sigma^2    \int_{(p,q)} \left(\frac{\dir{}^2}{\dir \sigma^2} h(\sigma,\|p-q\|) \right) \dir{}\textsf{M}_{\mu,\delta_x}(p,q)
\leq
\sigma^2 \left((18/e^{3})/\sigma^2\right) 2
=36/e^3  
\]
 
Then we complete the proof using the upper bound of each item of $\frac{\dir{}^2}{\dir \sigma^2} F(\sigma)$
\begin{align*}
\frac{\dir{}^2}{\dir \sigma^2} F(\sigma) 
& = 
 \sigma^2 \frac{\dir{}^2}{\dir \sigma^2} f_x(\sigma)  +  4 \sigma \frac{\dir{}}{\dir \sigma} f_x(\sigma)  + 2 f_x(\sigma) - 2 \ell
 \\ & \leq 36/e^3 + 16/e + 4 - 2(18/e^3+8/e+2) = 0.  \qedhere
\end{align*}
\end{proof}

\paragraph{Lipschitz in $m_0$ for $\dC_{\mu,m_0}$.}
We show that there is no Lipschitz property for $\dC_{\mu,m_0}$, with respect to $m_0$ that is independent of the measure $\mu$.  
Consider a measure $\mu_P$ for point set $P \subset \mathbb{R}$ consisting of two points at $a = 0$ and at $b=\Delta$.
Now consider $\dC_{\muP,m_0}(a)$.  When $m_0 \leq 1/2$ then $\dC_{\muP,m_0}(a) =0$ is constant.  But for $m_0 = 1/2 + \alpha$ for $\alpha>0$, then $\dC_{\muP,m_0}(a) = \alpha\Delta/(1/2+\alpha)$ and 
$\frac{\dir{}}{\dir{m_0}} \dC_{\muP,m_0}(a) = \frac{\dir{}}{\dir{\alpha}} \dC_{\muP,\frac{1}{2}+\alpha}(a) = \frac{(1/2 + 2\alpha)\Delta}{(1/2 + \alpha)^2}$, which is maximized as $\alpha$ approaches $0$ with an infimum of $2\Delta$.  
If $n-1$ points are at $b$ and $1$ point at $a$, then the infimum of the first derivative of $m_0$ is $n\Delta$.  
Hence for a measure $\mu_P$ defined by a point set, the infimum of $\frac{\dir{}}{\dir{m_0}} \dC_{\muP,m_0}(a)$ and, hence a lower bound on the Lipschitz constant is $n \Delta_P$ where $\Delta_P = \max_{p,p' \in P} \|p-p'\|$.

\section{Algorithmic and Approximation Observations}
\label{sec:algo}

\paragraph{Kernel coresets.}
The kernel distance is robust under random samples~\cite{JoshiKommarajuPhillips2011}.  
Specifically, if $Q$ is a point set randomly chosen from $\mu$ of 
size $O((1/\eps^2) (d + \log(1/\delta))$ then 
$\| \kde_\mu - \kde_Q \|_{\infty} \leq \eps$
with probability at least $1-\delta$.  We call such a subset $Q$ and $\eps$-kernel sample of $(\mu, K)$.  
Furthermore, it is also possible to construct $\eps$-kernel samples $Q$ with even smaller size of $|Q| = O(((1/\eps) \sqrt{\log (1/\eps \delta)})^{2d/(d+2)})$~\cite{Phillips2013}; in particular in $\Rspace^2$ the required size is $|Q| = O((1/\eps) \sqrt{\log(1/\eps \delta)})$.  
Exploiting the above constructions, recent work~\cite{big-kde} builds a data structure to allow for efficient approximate evaluations of $\kde_P$ where $|P| = 100{,}000{,}000$.  

These constructions of $Q$ also immediately imply that $\|(\dK_\mu)^2 - (\dK_Q)^2\|_\infty \leq 4\eps$ since $(\dK_\mu(x))^2 = \kappa(\mu,\mu) + \kappa(x,x) - 2\kde_\mu(x)$, and both the first and third term incur at most $2 \eps$ error in converting to $\kappa(Q,Q)$ and $2\kde_Q(x)$, respectively (see Lemma \ref{lem:KD-samp}).  
Thus, an $(\eps^2/4)$-kernel sample $Q$ of $(\mu,K)$ implies that $\|\dK_\mu - \dK_Q\|_\infty \leq \eps$ (see Lemma \ref{lem:KD-samp-nosq}).

This implies algorithms for geometric inference on enormous noisy data sets.  
Moreover, if we assume an input point set $Q$ is drawn iid from some underlying, but unknown distribution $\mu$, we can bound approximations with respect to $\mu$.

\begin{corollary}
Consider a measure $\mu$ defined on $\mathbb{R}^d$, a kernel $K$, and a parameter $\eps \leq R(5 + 4/\alpha^2)$.  
We can create a coreset $Q$ of size $|Q| = O(((1/\eps^2) \sqrt{\log (1/\eps \delta)})^{2d/(d+2)})$ or randomly sample $|Q| = O((1/\eps^4)(d + \log(1/\delta)))$ points so, with probability at least $1-\delta$, any sublevel set $(\dK_\mu)^\eta$ is homotopy equivalent to $(\dK_Q)^r$ for $r \in [4\eps/\alpha^2, R-3\eps]$ and $\eta \in (0, R)$.  
\end{corollary}
\begin{proof}
Those bounds are obtained by constructing an $(\eps^2/4)$-kernel sample~\cite{JoshiKommarajuPhillips2011,Phillips2013}, which guarantees $\|\dK_\mu - \dK_Q\|_\infty \leq \eps$ via Lemma \ref{lem:KD-samp-nosq}.  
Then since $\eps \leq R/(5 + 4/\alpha^2)$, with Theorem \ref{thm:reconstruct} any sublevel set $(\dK_\mu)^\eta$ is homotopy equivalent to $(\dK_Q)^r$.  
\end{proof}

\paragraph{Stability of persistence diagrams.}
Furthermore, the stability results on persistence diagrams \cite{Cohen-SteinerEdelsbrunnerHarer2007} hold for kernel density estimates and kernel distance of $\mu$ and $Q$ (where $Q$ is a coreset of $\mu$ with the same size bounds as above).  If $\|f - g\|_\infty \leq \eps$, then $d_B(\dgmD{f}, \dgmD{g}) \leq \eps$, where $d_B$ is the bottleneck distance between persistence diagrams.  Combined with the coreset results above, this immediately implies the following corollaries.  

\begin{corollary}
Consider a measure $\mu$ defined on $\mathbb{R}^d$ and a kernel $K$.  
We can create a core set $Q$ of size $|Q| = O(((1/\eps) \sqrt{\log (1/\eps \delta)})^{2d/(d+2)})$ or randomly sample $|Q| = O((1/\eps^2)(d + \log(1/\delta)))$ points which will have the following properties with probability at least $1-\delta$.  
\begin{itemize}
\item $d_B(\dgmD{\kde_\mu}, \dgmD{\kde_Q}) \leq \eps$.
\item $d_B(\dgmD{(\dK_\mu)^2}, \dgmD{(\dK_Q)^2}) \leq \eps$.
\end{itemize}
\end{corollary}

\begin{corollary}
Consider a measure $\mu$ defined on $\mathbb{R}^d$ and a kernel $K$.  
We can create a coreset $Q$ of size $|Q| = O(((1/\eps^2) \sqrt{\log (1/\eps \delta)})^{2d/(d+2)})$ or randomly sample $|Q| = O((1/\eps^4)(d + \log(1/\delta)))$ points which will have the following property with probability at least $1-\delta$.  
\begin{itemize}
\item $d_B(\dgmD{\dK_\mu}, \dgmD{\dK_Q}) \leq \eps$.
\end{itemize}
\end{corollary}

Another bound was independently derived to show an upper bound on the size of a random sample $Q$ such that $d_B(\dgmD{\kde_{\mu_P}}, \dgmD{\kde_Q}) \leq \eps$ in \cite{BFLRSW13};  this can, as above, also be translated into bounds for $\dgmD{(\dK_Q)^2}$ and $\dgmD{\dK_Q}$.  This result assumes $P \subset [-C,C]^d$ and is parametrized by a bandwidth parameter $h$ that retains that $\int_{x \in \mathbb{R}^d} K_h(x,p) \dir{x} = 1$ for all $p$ using that $K_1(\|x-p\|) =K(x,p)$ and $K_h(\|x-p\|) = \frac{1}{h^d} K_1(\|x-p\|^2/h)$.  This ensures that $K(\cdot, p)$ is $(1/h^d)$-Lipschitz and that $K(x,x) = \Theta(1/h^d)$ for any $x$.  Then their bound requires
$|Q| = O(\frac{d}{\eps^2 h^d} \log(\frac{C d }{\eps \delta h}))$ random samples.  

To compare directly against the random sampling result we derive from Joshi \etal~\cite{JoshiKommarajuPhillips2011}, for kernel $K_h(x,p)$ then $\|\kde_{\mu_P} - \kde_Q\|_\infty \leq \eps K_h(x,x) = \eps / h^d$.  Hence, our analysis requires $|Q| = O((1/\eps^2 h^{2d}) (d + \log(1/\delta)))$, and is an improvement when $h = \Omega(1)$ or $C$ is not known or bounded, as well as in some other cases as a function of $\eps$, $h$, $\delta$, and $d$.

\section{Discussion}
\label{app:discuss}
We mention here a few other interesting aspects of our results and observations about topological inference using the kernel distance.  They are related to how the noise parameter $\sigma$ affects the idea of scale, and a few more experiments, including with alternate kernels.

\subsection{Noise and Scale}
Much of geometric and topological reconstruction grew out of the desire to understand shapes at various scales.  A common mechanism is offset based; e.g., $\alpha$-shapes~\cite{Edelsbrunner1993} represent the scale of a shape with the $\alpha$ parameter controlling the offsets of a point cloud.  
There are two parameters with the kernel distance: $r$ controls the offset through the sublevel set of the function, and $\sigma$ controls the noise. 
We argue that any function which is robust to noise must have a parameter that controls the noise (e.g. $\sigma$ for $\dK_\mu$ and $m_0$ for $\dC_{\mu,m_0}$).
Here $\sigma$ clearly defines some sense of scale in the setting of density estimation~\cite{Sil86} and has a geometrical interpretation, while $m_0$ represents a fraction of the measure and is hard to interpret geometrically, as illustrated by the lack of a Lipschitz property for $\dC_{\mu,m_0}$ with respect to $m_0$.  

There are several experiments below, in Section \ref{sec:exp}, 
from which several insights can be drawn.  
One observation is that even though there are two parameters $r$ and $\sigma$ that control the scale, the interesting values typically have $r$ very close to $\sigma$.  
Thus, we recommend to first set $\sigma$ to control the scale at which the data is studied, and then explore the effect of varying $r$ for values near $\sigma$.  
Moreover, not much structure seems to be missed by not exploring the space of both parameters; 
Figure \ref{fig:gamma-v-sigma} shows that fixing one (of $r$ and $\sigma$) and varying the other can provide very similar superlevel sets.  
However, it is possible instead to fix $r$ and explore the persistent topological features in the data \cite{EdelsbrunnerLetscherZomorodian2002,EdelsbrunnerHarer2008} (those less affected by smoothing) by varying $\sigma$.  
On the other hand, it remains a challenging problem to study two parameter persistent homology \cite{CarlssonZomorodian2007,CarlssonSinghZomorodian2009} under the setting of kernel distance (or kernel density estimate).

\subsection{Experiments}
\label{sec:exp}

We consider measures $\mu_P$ defined by a point set $P \subset \Rspace^2$.  
To experimentally visualize the structures of the superlevel sets of kernel density estimates, or equivalently sublevel sets of the kernel distance, we do the simplest thing and just evaluate $\dK_\muP$ at every grid point on a sufficiently dense grid.

\paragraph{Grid approximation.}
Due to the $1$-Lipschitz property of the kernel distance, well chosen grid points have several nice properties.  
We consider the functions up to some resolution parameter $\eps > 0$, consistent with the parameter used to create a coreset approximation $Q$.  
Now specifically, consider an axis-aligned grid $G_{\eps,d}$ with edge length $\eps/2\sqrt{d}$ so no point $x \in \Rspace^d$ is further than $\eps$ from some grid point $g \in G_{\eps,d}$.  
Since $K(x,y) \leq \eps$ when $\|x-y\| \geq 2 \sigma^2 \ln(\sigma^2/\eps) = \delta_{\eps,\sigma}$, we only need to consider grid points $g \in G_{\eps,d}$ which are within $\delta_{\eps,\sigma}$ of some point $p \in P$ (or $q \in Q$, of coreset $Q$ of $P$)~\cite{JoshiKommarajuPhillips2011,big-kde}.  This is at most $(2\sqrt{d} / \eps)^d (2 \delta_{\eps,\sigma})^d = O((\sigma^2 \log (\eps/d) /\eps)^d)$ grid points total for $d$ a fixed constant.  
Furthermore, due to the $1$-Lipschitz property of $\dK_P$, when considering a specific level set at $r$
\begin{itemize}
\item a point $x$ such that $\dK_P(x) \leq r - \eps$ is no further than $\eps$ from some $g \in G$ such that $\dK_P(g) \leq r$, and 
\item every ball $B_{\eps}(x)$ centered at some point $x \in \Rspace^d$ of radius $\eps$ so that all $y \in B_{\eps}(x)$ has $\dK_P(y) \leq r$ has some representative point $g \in G_{\eps,d}$ such that $g \in B_{\eps}(x)$, and hence $\dK_P(g) \leq r$. 
\end{itemize} 
Thus ``deep" regions and spatially thick features are preserved, however thin passageways or layers that are near the threshold $r$, even if they do not correspond to a critical point, may erroneously become disconnected, causing phantom components or other topological features.  However, due to the Lipschitz property, these can be different from $r$ by at most $\eps$, so the errors will have small deviation in persistence.

\paragraph{Varying parameter $r$ or $\sigma$.}
\label{sec:exp}

\begin{figure}
\begin{center}
\includegraphics[width=.4\linewidth]{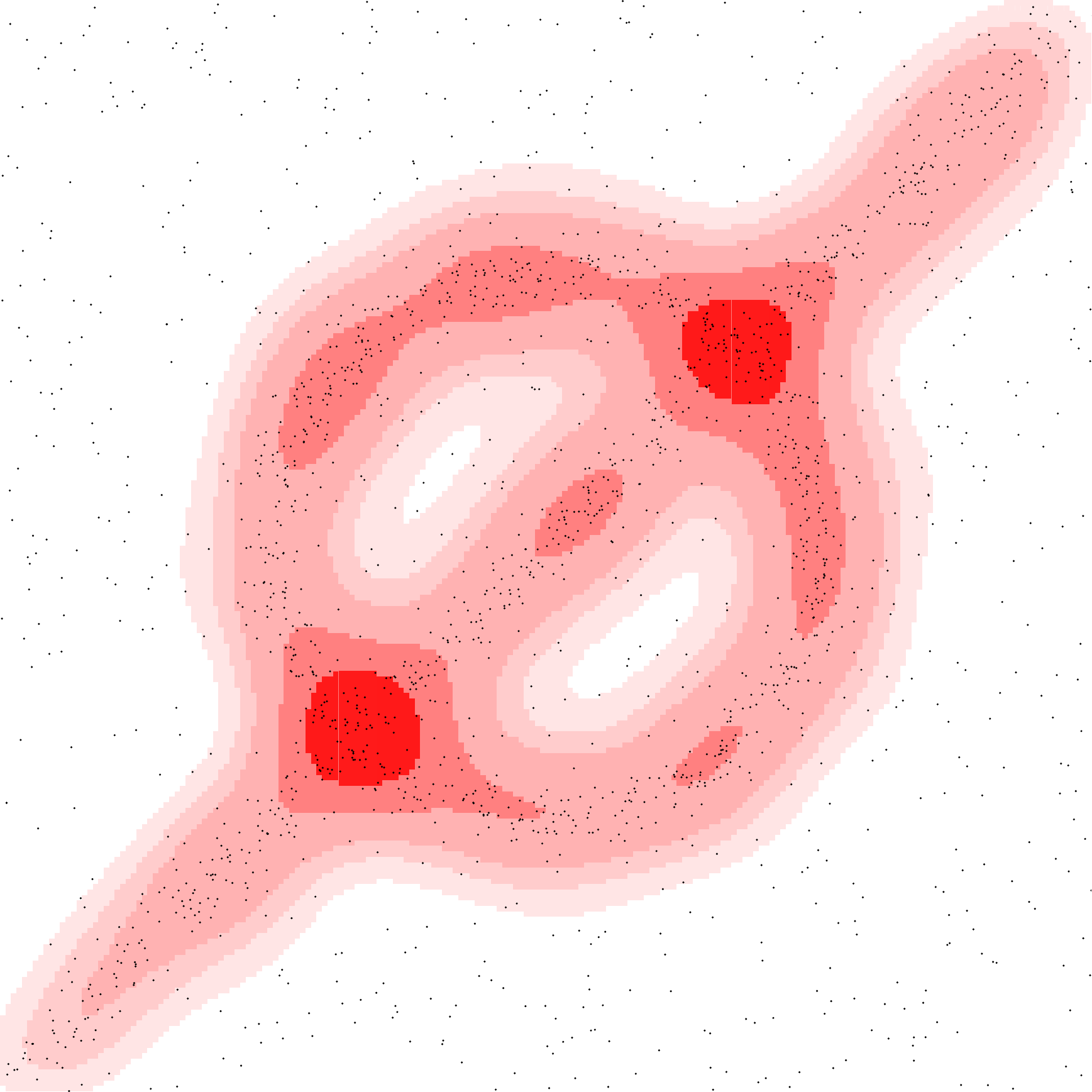}
\hspace{10mm}
\includegraphics[width=.4\linewidth]{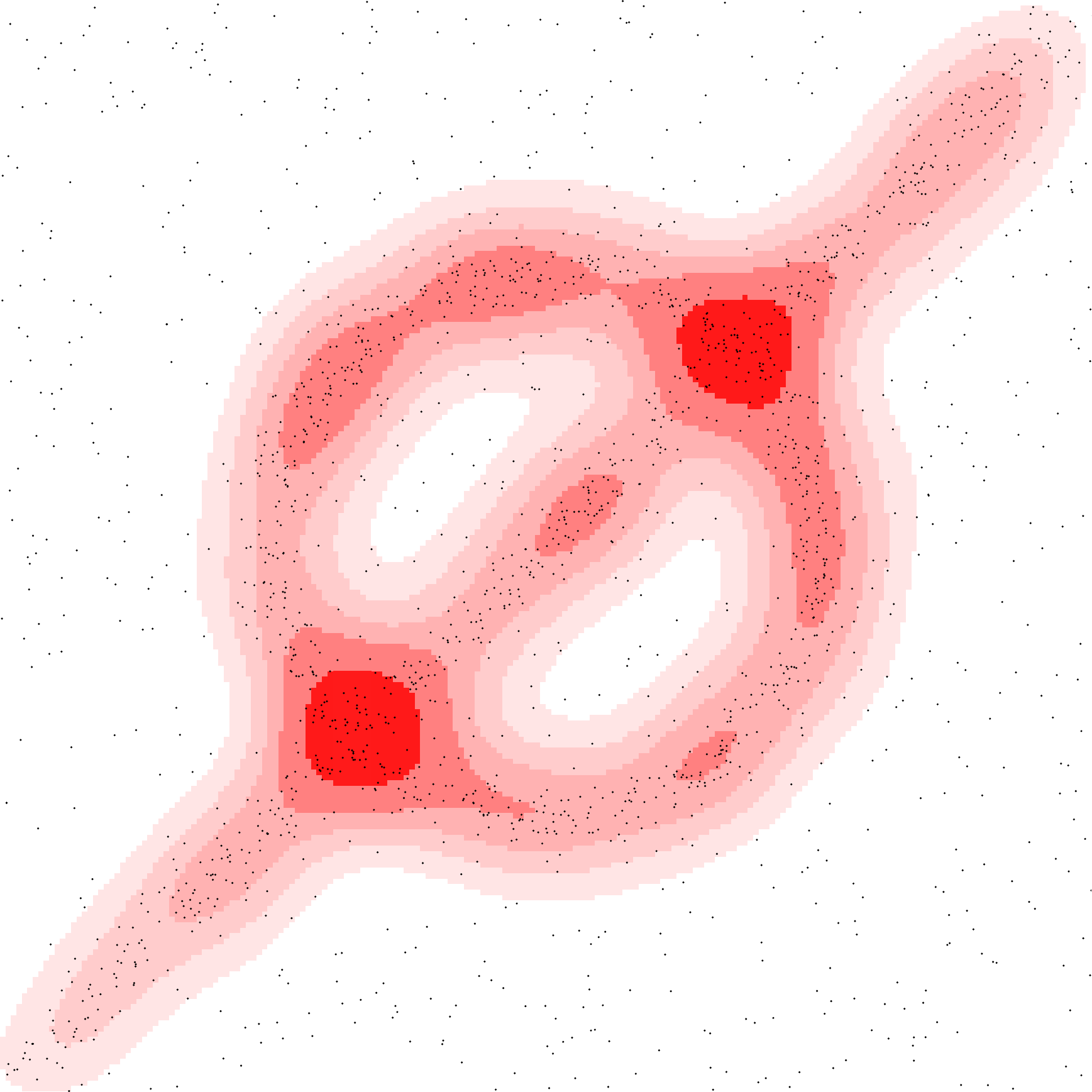}
\end{center}

\caption{\label{fig:gamma-v-sigma} \small \sffamily
Sublevel sets for the kernel distance while varying the isolevel $\gamma$, for fixed $\sigma$ (left) and for fixed isolevel $\gamma$ but variable $\sigma$ (right), with Gaussian kernel.  The variable values of $\sigma$ and $\gamma$ are chosen to make the plots similar.}  
\end{figure}

We demonstrate the geometric inference on a synthetic dataset in $[0,1]^2$ where $900$ points are chosen near a circle centered at $(0.5,0.5)$ with radius $0.25$ or along a line segment from $(0,0)$ to $(1,1)$.  Each point has Gaussian noise added with standard deviation $0.01$.  The remaining $1100$ points are chosen uniformly from $[0,1]^2$.  
We use a Gaussian kernel with $\sigma = 0.05$.  
Figure \ref{fig:gamma-v-sigma} shows 
(left) various sublevel sets $\gamma \in \Gamma$ for the kernel distance at a fixed $\sigma = 0.05$ and (right) various superlevel sets for a fixed $\gamma = 0.04853$, but various values of $\sigma \in \Sigma$, where

\begin{align*}
\Gamma &= [0.05005, 0.04979, 0.04954, 0.04904, 0.04853] \text{ and }
\\
\Sigma &= [0.0485, 0.0489, 0.0492, 0.0495, 0.05].
\end{align*}

\noindent
This choice of $\Gamma$ and $\Sigma$ were made to highlight how similar the isolevels can be.

\begin{figure}
\includegraphics[width=.24\linewidth]{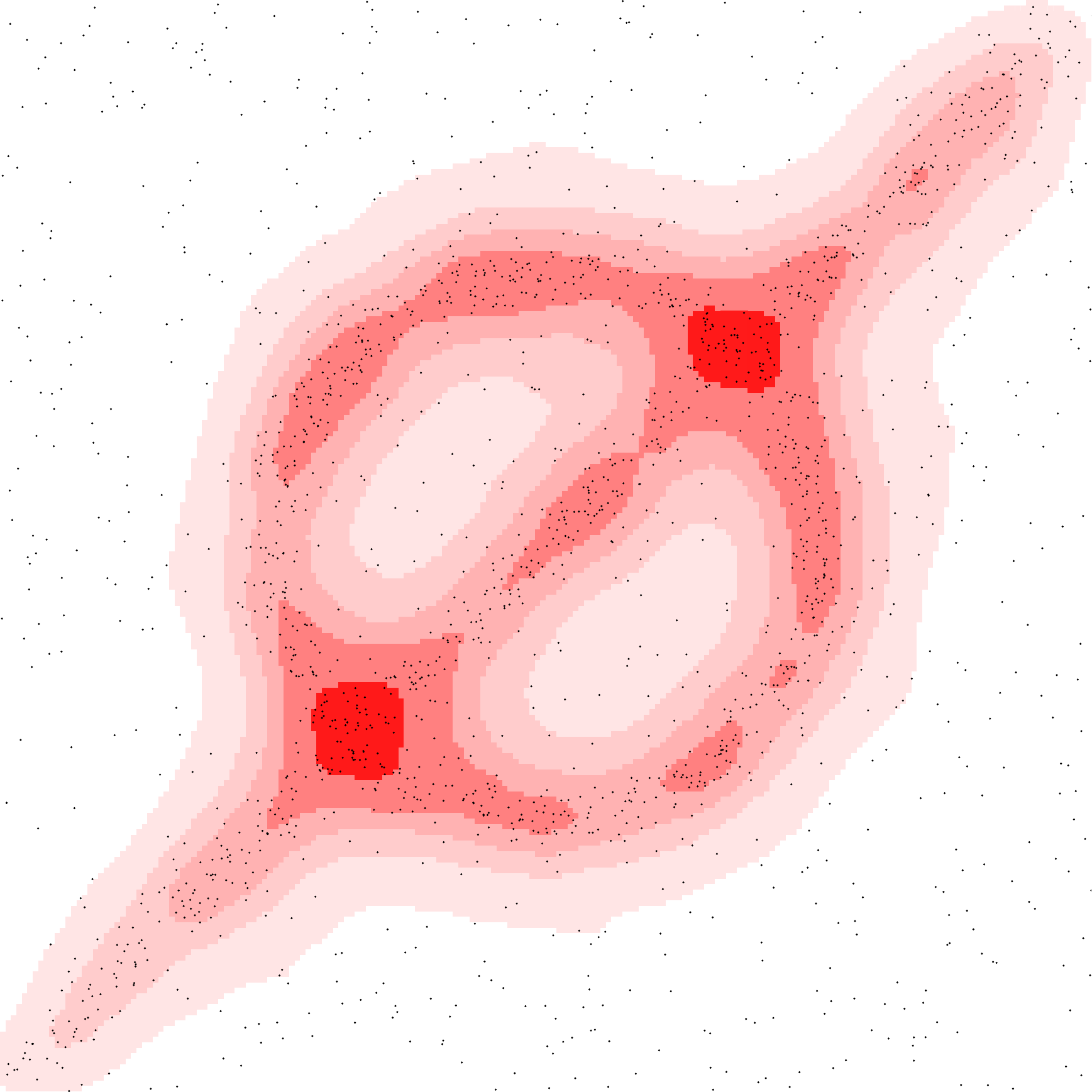}
\hspace{0.3mm}
\includegraphics[width=.24\linewidth]{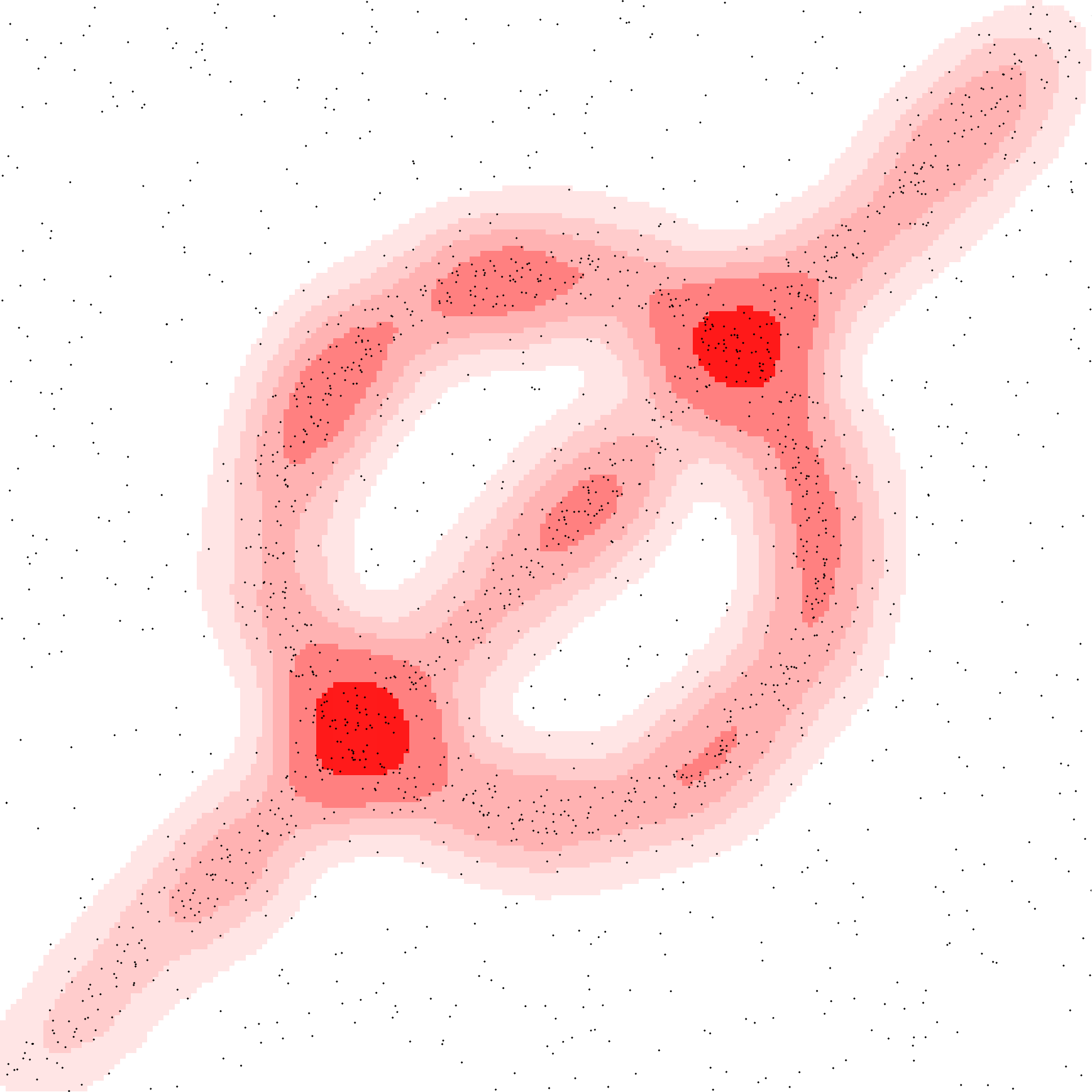}
\hspace{0.3mm}
\includegraphics[width=.24\linewidth]{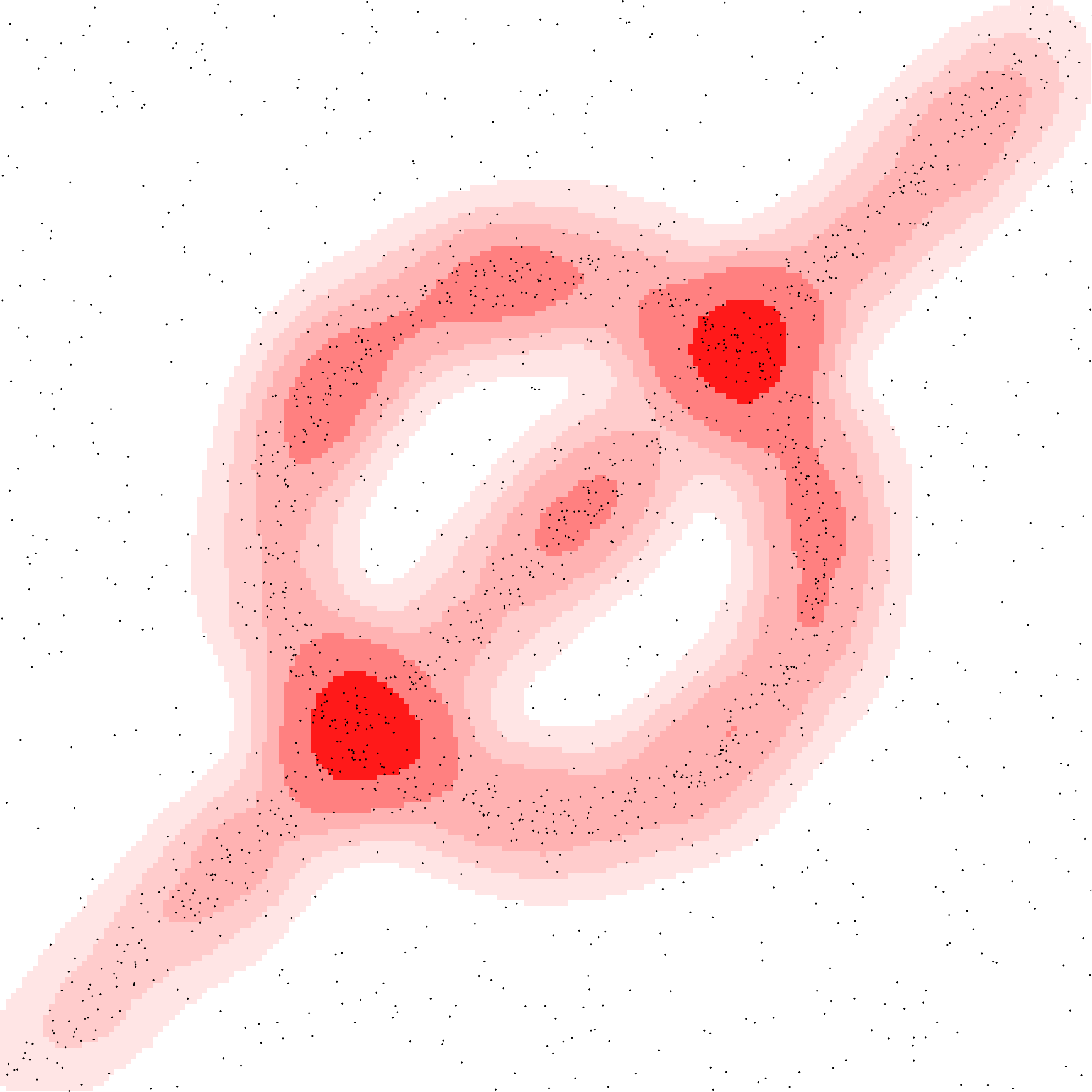}
\hspace{0.3mm}
\includegraphics[width=.24\linewidth]{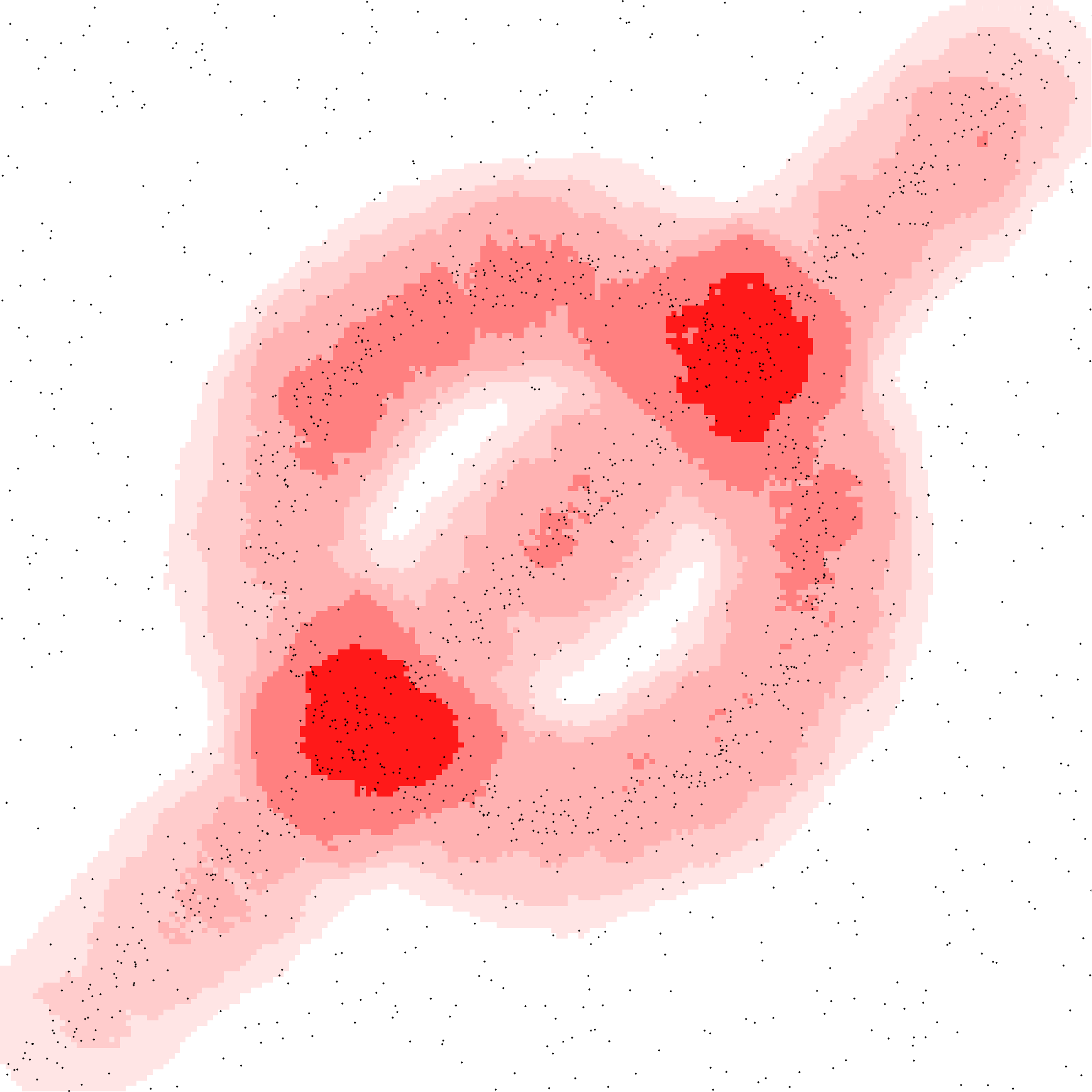}

\caption{\label{fig:TEB}  
\small \sffamily
Alternate kernel density estimates for the same dataset as Figure \ref{fig:gamma-v-sigma}.  From left to right, they use the Laplace, triangle, Epanechnikov, and the ball kernel.}
\end{figure}

\paragraph{Alternative kernels.}
We can choose kernels other than the Gaussian kernel in the kernel density estimate, for instance 
\begin{itemize} \denselist
\item the Laplace kernel $K(p,x) = \exp(-2\|x-y\|/\sigma)$, 
\item the triangle kernel $K(p,x) = \max \{0, 1-\|x-y\|/\sigma\}$, 
\item the Epanechnikov kernel $K(p,x) = \max \{0, 1 - \|x-y\|^2/\sigma^2\}$, or 
\item the ball kernel ($K(p,x) = \{1 \text{ if } \|p-x\| \leq \sigma \text{; o.w. } 0\}$.   
\end{itemize}
Figure \ref{fig:TEB} chooses parameters to make them comparable to the Figure \ref{fig:gamma-v-sigma}(left).  
Of these only the Laplace kernel is \emph{characteristic}~\cite{SGFSL10} making the corresponding version of the kernel distance a metric.  Investigating which of the above reconstruction theorems hold when using the Laplace or other kernels is an interesting question for future work.

Additionally, normal vector information and even $k$-forms can be used in the definition of a kernel~\cite{glaunesthesis,Vaillant2005,DurrlemanPennecTrouve2008,
DurrlemanPennecTrouve2007,GlaunesJoshi2006,JoshiKommarajuPhillips2011}; this variant is known as the \emph{current distance}.  In some cases it retains its metric properties and has been shown to be very useful for shape alignment in conjunction with medical imaging.

\subsection{Open Questions}
\label{sec:open}
This work shows it is possible to prove formal reconstruction results using kernel density estimates and the kernel distance.  But it also opens many interesting questions.  

\begin{itemize}
\item For what other types of kernels can we show reconstruction bounds?  The Laplace and triangle kernels are natural choices.  For both the coresets results match those of the Gaussian kernel.  The kernel distance under the Laplace kernel is also a metric, but is not known to be for the triangle kernel.  Yet, the triangle kernel would be interesting since it has bounded support, and may lend itself to easier computation.  

\item The power distance construction in Section \ref{sec:powerdistance} requires a point $\hat p_+$, which approximates the point with minimum kernel distance.  This is intuitively because it is possible to construct a point set $P$ (say points lying on a circle with no points inside) such that the point $p_+ \in \mathbb{R}^d$ which minimizes the kernel distance and maximizes the kernel density estimate is far from any point in the point set.  
For one, can $\hat p_+$ be constructed efficiently without dependence on $\beta_P$ or $\Lambda_P/\sigma$?  

But more interestingly, can we generally approximate the persistence diagram without creating a simplicial complex on a subset of the input points?  We do describe some bounds on using a grid-based technique in Section \ref{sec:exp}, but this is also unsatisfying since it essentially requires a low-dimensional Euclidean space.  

\item Since $\dK_\mu$ is Lipschitz in $x$ and $\sigma$, it may make sense to understand the simultaneous stability of both variables.  What is the best way to understand persistence over both parameters?  

\item We provided some initial bound comparing the kernel distance under the Gaussian kernel and the Wasserstein $2$ distance.  Can we show that under our choice of normalization that $D_K(\mu,\nu) \leq W_2(\mu,\nu)$, unconstrained?  More generally, how does the kernel distance under other kernels compare with other forms of Wasserstein and other distances on measures?  
\end{itemize}

\paragraph{Acknowledgements}
The authors thank Don Sheehy, Fr\'{e}d\'{e}ric Chazal and the rest of the Geometrica group at  INRIA-Saclay for enlightening discussions on geometric and topological reconstruction.  
We also thank Don Sheehy for personal communications regarding the power distance constructions, 
and Yusu Wang for ideas towards Lemma \ref{lem:monotonicity}.  Finally, we are also indebted to the anonymous reviewers for many detailed suggestions leading to improvements in results and presentation.

\bibliographystyle{plain}
\bibliography{kernel-refs}


\newpage
\appendix

\section{Details on Distance-Like Properties of Kernel Distance}
\label{app:dist-like}

We provide further details on distance-like properties of the kernel distance.  

\subsection{Semiconcave Properties of Kernel Distance}

We also note that semiconcavity follows quite naturally and simply in the RKHS $\Eu{H}_K$ for $\dK_\mu$.  

\begin{lemma}
\label{lemma:RKHS-concave}
$(\dK_\mu)^2$ is $1$-semiconcave in $\Eu{H}_K$: the map $x \mapsto (\dK_\mu(x))^2 - \|\phi(x)\|_{\Eu{H}_K}^2$ is concave. 
\end{lemma}
\begin{proof}
We can write 
\[
(\dK_{\mu}(x))^2 = (D_K(\mu,x))^2 
= 
\kappa(\mu,\mu) + \kappa(x,x) - 2 \kappa(\mu,x)
=
\|\Phi(\mu)\|^2_{\Eu{H}_K} + \|\phi(x)\|^2_{\Eu{H}_K} - 2 \|\Phi(\mu) - \phi(x)\|^2_{\Eu{H}_K}.  
\]
Now 
\[
(\dK_\mu(x))^2 - \|\phi(x)\|_{\Eu{H}_K}^2 
=
\|\Phi(\mu)\|^2_{\Eu{H}_K} - 2 \|\Phi(\mu) - \phi(x)\|^2_{\Eu{H}_K}.  
\]
Since the above is twice-differentiable, we only need to show that its twice-differential is non-positive. 
By definition, for a fixed $\mu$, $\Phi(\mu)$ and $\|\Phi(\mu)\|^2_{\Eu{H}_K}$ are both constant. 
Suppose $\Phi(\mu) = c_1$ and $\|\Phi(\mu)\|^2_{\Eu{H}_K} = c_2$,
we have $(d_\mu(x))^2 - \|\phi(x)\|_{\Eu{H}_K}^2  = c_2 - \|c_1 - \phi(x)\|^2_{\Eu{H}_K}$.
Since the RKHS $\Eu{H}_K$ is a vector space with well-defined norm $\| \cdot \|_{\Eu{H}_K}$, 
the above is a concave parabolic function.  
\end{proof}

However, this semiconcavity in $\Eu{H}_K$ is not that useful. 
For unit weight elements $x,y \in \Rspace^d$, an element $s_\alpha$ such that $\phi(s_\alpha) = \alpha \phi(y) + (1-\alpha) \phi(x)$ is a weighted point set with a point at $x$ with weight $(1-\alpha)$ and another at $y$ with weight $\alpha$. 
Lemma \ref{lemma:RKHS-concave} only implies that 
$(d_K(s_\alpha))^2 - \|\phi(s_\alpha)\|^2_{\Eu{H}_K} \leq \alpha((d_K(x))^2 - \|\phi(x)\|^2_{\Eu{H}_K}) + (1-\alpha)((d_K(y))^2 - \|\phi(y)\|^2_{\Eu{H}_K})$.

\subsection{Kernel Distance is Proper}

We use two more general, but equivalent definitions of a proper map.  
Definition (i): A continuous map $f: \Xspace \to \Yspace$ between two topological spaces is \emph{proper} if and only if the inverse image of every compact subset in $\Yspace$ is compact in $\Xspace$ (\cite{Lee2000}, page 84; \cite{Lee2003}, page 45). 
Definition (ii): a continuous map $f: \Xspace \to \Yspace$ between two topological manifolds is proper if and only if for every sequence $\{p_i\}$ in $\Xspace$ that escapes to infinity, $\{f(p_i)\}$ escapes to infinity in $\Yspace$ (\cite{Lee2003}, Proposition 2.17).
Here, for a topological space $\Xspace$, a sequence $\{p_i\}$ in $\Xspace$ \emph{escapes to infinity} if for every compact set $G \subset \Xspace$, there are at most finitely many values of $i$ for which $p_i \in G$ (\cite{Lee2003}, page 46).

\begin{lemma}[Lemma \ref{lem:proper}] 
$d^K_{\mu}$ is proper. 
\label{lem:proper-app}
\end{lemma}
\begin{proof}
To prove that $d^K_{\mu}$ is proper, we prove the following two claims: 
(a) A continuous function $f: \Rspace^d \to [0, c)$ (where c is a constant) is proper, if for any sequence $\{x_i\}$ in $\Rspace^d$ that escapes to infinity, the sequence 
$\{f(x_i)\}$ tends to $c$ (approaches $c$ in the limit);
(b) Let $f: = d^K_{\mu}$ and one needs to show that for any sequences $\{x_i\}$ that escapes to infinity, the sequence $\{f(x_i)\}$ tends to $c_{\mu}$; or equivalently, 
$\kappa(\mu, x_i)$ tends to $0$.

We prove claim (a) by proving its contrapositive.   
If a continuous function $f: \Rspace^d \to [0, c)$ is not proper, then there exists a sequence $\{x_i\}$ in $\Rspace^d$ that escapes to infinity, such that the sequence $\{f(x_i)\}$ does not tend to $c$. 
Suppose $f$ is not proper, this implies that there exists a constant $b < c$ such that $f^{-1}[0,b]$ is not compact (based on properness definition (i)) and therefore either not closed or unbounded. 
We first show that $A: = f^{-1}[0,b]$ is closed. 
We make use of the following theorem (\cite{KolmogorovFominSilverman1975}, page 88, Theorem 10'):
A mapping $f$ of a topological space $\Xspace$ into a topological space $\Yspace$ is continuous if and only if the pre-image $f^{-1}(F)$ of every closed set $F \subset \Yspace$ is closed in $\Xspace$.  
Since $f$ is continuous, it implies that the pre-image of every closed set $[a,b] \subset  R$ is closed in $\Rspace^d$. Therefore, $A$ is closed, therefore it must be unbounded.
Since every unbounded sequence contains a monotone subsequence
that has either $+\infty$ or $-\infty$ as a limit, 
therefore $A$ contains a subsequence $S := \{x_i\}$ that tends to an infinite limit. 
In addition, as elements in $S$ escapes to infinity, $\{f(x_i)\}$ tends to $b$ and does not tend to $c$.
Therefore (a) holds by contraposition. 

To prove claim (b), we need to show that for any sequence $\{x_i\}$ that escapes to infinity, $\kappa(\mu, x_i)$ tends to $0$.  
For each $x_i$, define a radius $r_i = \|x_i - 0\|/2$ and define a ball $B_i$ that is centered at the origin $0$ and has radius $r_i$.  As $x_i$ goes to infinity, $r_i$ increases until for any fixed arbitrary $\eps > 0$,  we have $\int_{p \in B_i} \mu(p) \dir p \geq 1-\eps/2\sigma^2$ and thus $\int_{p \in \Rspace^d \setminus B_i} \dir{\mu(p)} \leq \eps/2\sigma^2$.  Furthermore, let $p_i = \arg\min_{p \in B} \|p - x_i\|$, so $\|x_i - p_i\| = r_i$.  Thus also as $x_i$ goes to infinity, $r_i$ increases until for any $\eps > 0$ we have $K(p_i,x_i) \leq \eps/2$.  
We now decompose $\kappa(\mu,x_i) = \int_{p \in B_i} K(p,x_i) \dir{\mu(p)} + \int_{q \in \Rspace^d \setminus B_i} K(q,x_i) \dir{\mu(q)}$.  
Thus for any $\eps >0$, as $x_i$ goes to infinity, the first term is at most $\eps/2$ since all $K(p,x_i) \leq K(p_i,x_i) \leq \eps/2$ and the second term is at most $\eps/2$ since $K(q,x) \leq \sigma^2$ and $\int_{q \in \Rspace^d \setminus B_i} \mu(q) \dir q \leq \eps / 2 \sigma^2$.  
Since these results hold for all $\eps$, as $x_i$ goes to infinity and $\eps$ goes to $0$, 
$\kappa(\mu,x_i)$ goes to $0$.  

Combine (a) with (b) and the fact that $\dK_{\mu}$ is a continuous (in fact, Lipschitz) function, we obtained the properness result. 
\end{proof}

\section{$\eps$-Approximation of the Kernel Distance}
\label{app:KD-approx}

Here we make explicit the way that an $\eps$-kernel sample approximated the kernel distance.  
Recall that if $Q$ is an $\eps$-kernel sample of $\mu$, then $\|\kde_\mu - \kde_\muQ \| =  \max_{x \in \Rspace^d} | \kappa(\mu,x) - \kappa(\muQ,x) | \leq \eps$.

\begin{lemma}
\label{lem:KD-samp}
If $Q$ is an $\eps$-kernel sample of $\mu$, then $\| (\dK_\mu)^2 - (\dK_\muQ)^2 \|_\infty \leq 4\eps$. 
\end{lemma}
\begin{proof}
First expand $D_K(\mu,x)^2 =  \kappa(x,x) + \kappa(\mu,\mu)  - 2 \kappa(\mu,x) = \sigma^2 + \kappa(\mu,\mu)  - 2 \kappa(\mu,x)$.  
Replacing $\mu$ with $\muQ$, 
the first term is unaffected.  The second term is bounded, 
\begin{align*}
\kappa(\mu,\mu) &= \int_{(p,q)} K(p,q) \dir{\textsf{m}_{\mu,\mu}(p,q)}
=  
\int_p \left( \int_q K(p,q) \dir{\mu(q)} \right) \dir{\mu(p)}
\\ &=
\int_p \kde_{\mu}(p) \dir{\mu(p)}
\leq 
\int_p (\kde_{\muQ}(p) + \eps) \dir{\mu(p)}
\\ &=  
\int_p \kde_{\muQ}(p) \dir{\mu(p)} + \eps
= \int_p \left( \int_q K(p,q) \dir{\muQ(q)} \right) \dir{\mu(p)} + \eps
\\ & = 
\kappa(\muQ, \mu) + \eps
\\ & \leq  \kappa(\muQ, \muQ) + 2 \eps. 
\end{align*}
Similar results hold by switching $\muQ$ with $\mu$ in the above inequality, that is, 
$\kappa(\muQ,\muQ) \leq \kappa(\mu,\mu)  + 2\eps$.   
And for the third term we have similar inequality, $|2\kappa(\mu, x) - 2\kappa(\muQ,x)| \leq 2\eps$.
Combining all three terms,  we have the desired result: 
  $| D_K(\mu,x)^2 - D_K(\muQ,x)^2 | \leq 4\eps$. 
\end{proof}

\begin{lemma}
\label{lem:KD-samp-nosq}
If $Q$ is an $(\eps^2/4)$-kernel sample of $\mu$, then $\| \dK_\mu - \dK_\muQ \|_\infty \leq \eps$. 
\end{lemma}
\begin{proof}
By Lemma \ref{lem:KD-samp} this condition on $Q$ implies that $\| (\dK_\mu)^2 - (\dK_\muQ)^2\|_\infty \leq \eps^2$.  
We then use a basic fact for values $\eps \geq 0$ and $\gamma \geq 0$.  
\\
$\bullet$ $\sqrt{\gamma^2 + \eps^2} \leq \gamma + \eps$.  This follows since
$(\gamma + \eps)^2 = \gamma^2 + \eps^2 + 2 \gamma\eps \geq \gamma^2 + \eps^2$.

We now prove the main result as an upper and lower bound using for any $x \in \mathbb{R}^d$.  
We first use $\gamma = \dK_\mu(x) \geq 0$ and expand $\dK_\muQ(x)$ to obtain
\[
\dK_\muQ(x) = \sqrt{(\dK_\muQ(x))^2} \leq \sqrt{(\dK_\mu(x))^2 + \eps^2} \leq \dK_\mu(x) + \eps.  
\]

Now we use $\gamma = \dK_\muQ(x) \geq 0$ and expand $\dK_\mu(x)$ to obtain
\[
\dK_\mu(x) = \sqrt{(\dK_\mu(x))^2} \leq \sqrt{(\dK_\muQ(x))^2 + \eps^2} \leq \dK_\muQ(x) + \eps.  
\]
Hence for any $x \in \mathbb{R}^d$ we have $\dK_\mu(x) - \eps \leq \dK_\muQ(x) \leq \dK_\mu(x) + \eps$.  
\end{proof}

\section{Power Distance Constructions}
\label{app:power}
Recall we want to consider the following \emph{power distance} using $\dK_\mu$ (as weight) for a measure $\mu$ associated with a subset $P \subset \Rspace^d$ and metric $d(\cdot, \cdot)$ on $\Rspace^d$,
\[
\Dpow{P}(\mu,x) = \sqrt{\min_{p \in P} \left( {d(p,x)}^2 + {\dK_\mu(p)}^2 \right) }. 
\]
We consider a particular choice of the distance metric $d(p,x) = D_K(p,x)$ which leads to a kernel version of the power distance
\[
\Kpow{P}(\mu, x) = \sqrt{\min_{p \in P} \left( {D_K(p,x)}^2 + {\dK_\mu(p)}^2 \right) }.
\]

Recall that $\dK_\mu(x) = D_K(\mu,x)$.  In this section, we will always use the notation $D_K(\mu,\nu)$, and when $\mu$ or $\nu$ are points (e.g. $\mu$ is a Dirac mass at $p$ and $\nu$ is a Dirac mass at $q$), then we will just write $D_K(p,q)$.  This will be especially helpful when we apply the triangle inequality in several places.

\subsection{Kernel Power Distance on Point Set $P$}
\label{app:pow-justP}
Given a set $P$ defining a measure of interest $\mu_P$, it is of interest to consider if $\Kpow{P}(\muP,x)$ is multiplicatively bounded by $D_K(\muP,x)$.  Theorem \ref{thm:powK-low} shows that the lower bound holds.  In this section we try to provide a multiplicative approximation upper bound.  

Let $p^\star = \arg \min_{p \in P} \|p-x\|$.
We can start with Lemma \ref{lem:pK-up} which reduces the problem finding a multiplicative upper bound for $D_K(p^\star,x)$ in terms of $D_K(\muP,x)$.  
However, we are not able to provide very useful bounds, and they require more advanced techniques that the previous section.  In particular, they will only apply for points $x \in \mathbb{R}^d$ when $D_K(\muP,x)$ is large enough; hence not well-approximating the minima of $\dK_\mu$.  

For simplicity, we write $\dK_P(\cdot) = D_K(\muP,\cdot)$ as $D_K(P, \cdot)$. 

The difficult case is when $D_K(P,x)$ is very small, and hence $\kappa(P,P)$ is very small.  So we start by developing tools to upper bound $\kappa(P,P)$ using $\hat p = \arg \min_{p \in P} D_K(P,p)$, a point which only provides a worse approximation that $p^\star$.  

We first provide a general result in a Hilbert space (a refinement of a vector space \cite{Daume2004}), and then next apply it to our setting in the RKHS.

\begin{lemma}
\label{lem:vec-shrink}
Consider a set $V = \{v_1, \ldots, v_n\}$ of vectors in a Hilbert space endowed with norm $\|\cdot\|$ and inner product $\langle \cdot, \cdot \rangle$.  Let each $v_i$ have norm $\|v_i\| = \eta$.  
Consider weights $W = \{w_1, \ldots, w_n\}$ such that $w_i \geq 0$ and $\sum_{i=1}^n w_i = 1$.  Let $r = \sum_{i=1}^n w_i v_i$.  Let $\hat v = \arg \min_{v_i \in V} \|v_i - r\|$.
Then
 \[
 \|r\|^2 \leq \eta^2 - \|r - \hat v\|^2.
 \]
\end{lemma}
\begin{proof}
Recall elementary properties of inner product space: 
$\| x\| ^2 = \langle x, x\rangle$, 
$\langle ax, y\rangle = a \langle x, y \rangle$, 
$\langle x - y, x - y \rangle = \langle x, x \rangle + \langle y, y \rangle - 2 \langle x, y \rangle$. 
By definition of $\hat v$, 
for any $v_i \in V$, 
\begin{align*}
 \| v_i - r \|^2 \geq \| \hat v - r \|^2  
 \Rightarrow \langle v_i, v_i \rangle + \langle r, r \rangle - 2 \langle v_i, r \rangle 
\geq \langle \hat v, \hat v \rangle + \langle r, r \rangle - 2 \langle \hat v, r \rangle 
 \Rightarrow \langle v_i, r \rangle \leq \langle \hat v, r \rangle. 
\end{align*}
We can decompose $r$ (based on linearity of an inner product space) as 
\[
\|r\|^2 
= 
\langle r, r\rangle 
= 
\sum_{i=1}^n w_i \langle v_i, r \rangle
\leq
\sum_{i=1}^n w_i \langle \hat v, r \rangle
= 
\langle \hat v, r \rangle
= 
\frac{1}{2} ( \|r\|^2 + \|\hat v\|^2 - \|\hat v - r\|^2).
\]
The last inequality holds by $\|\hat v - r\|^2 = \|r\|^2 + \|\hat v\|^2 - 2 \langle \hat v, r\rangle$.  
Then since $\|\hat v\| = \eta$ we can solve for $\|r\|^2$ as 
\[
\|r\|^2 \leq \eta^2 - \|\hat v - r\|^2. \qedhere
\]
\end{proof}

\begin{lemma}
\label{lem:HK-shrink}
Let $\hat p = \arg \min_{p \in P} D_K(P, p)$, then
$\kappa(P,P) \leq \sigma^2 - D_K(P,\hat p)^2$.
\end{lemma}
\begin{proof}
Let $\phi_K : \Rspace^d \to \Eu{H}_K$ map points in $\Rspace^d$ to the reproducing kernel Hilbert space (RKHS) $\Eu{H}_K$ defined by kernel $K$.  This space has norm $\|P\|_{\Eu{H}_K} = \sqrt{\kappa(P,P)}$ defined on a set of points $P$ and inner product $\kappa(P,P)$.  Let $\Phi_K(P) = \frac{1}{|P|}\sum_{p \in P} \phi_K(p)$ be the representation of a set of points $P$ in $\Eu{H}_K$.  
Note that $D_K(P,Q) = \|\Phi_K(P) - \Phi_K(Q)\|_{\Eu{H}_K}$.  
We can now apply Lemma \ref{lem:vec-shrink} to $\{\phi_K(p)\}_{p \in P}$ with weights $w(p) = 1/|P|$ and $r = \Phi_K(P)$, and norm $\eta = \sigma$.  Hence
$\kappa(P,P) = \|P\|_{\Eu{H}_K}^2 \leq \sigma^2 - D_K(P,\hat p)^2$.
\end{proof}

\begin{lemma}
\label{lem:outa-root}
For any $s>0$ and any $x$, then $\sqrt{s^2 - x} \leq s-x/2s$.
\end{lemma}
\begin{proof}
We expand the square of the desired result
\[
(s^2 -x) \leq (s-x/2s)^2 = s^2 -x +x^2/4s^2.
\]
After subtracting $(s^2-x)$ from both sides, it is equivalent to $0 \leq x^2/4s^2$.  This holds since $x^2$ and $s$ are always nonnegative.  
\end{proof}

\begin{lemma}
\label{lem:sq-bnd}
$D_K(P,x) \geq D_K(p^\star, x)^2/C_\sigma$ for $C_\sigma = 2\sigma + 2$. 
\end{lemma}
\begin{proof}
Refer to Figure \ref{fig:kern-pow} for geometric intuition in this proof.  
Let $\nu_0$ be a measure that is $\nu_0(p) = 0$ for all $p \in \Rspace^d$; thus it has a norm $\kappa(\nu_0, \nu_0) = 0$.  
We can measure the distance from $\nu_0$ to $x$ and $P$, noting that $D_K(\nu_0,x) = \sqrt{\kappa(x,x)} = \sigma$ and $D_K(\nu_0, P) = \sqrt{\kappa(P,P)}$.  Thus by triangle inequality, Lemma \ref{lem:HK-shrink}, and Lemma \ref{lem:outa-root},   
\begin{align*}
D_K(P,x) 
& \geq 
D_K(\nu_0, x) - D_K(\nu_0,P) 
\\ &= 
\sigma - \sqrt{\kappa(P,P)}
\\ &\geq
\sigma - \sqrt{\sigma^2 - D_K(P,\hat p)^2}  
\\ & \geq
D_K(P,\hat p)^2/2\sigma.  
\end{align*}

\begin{figure}
\centering{\includegraphics{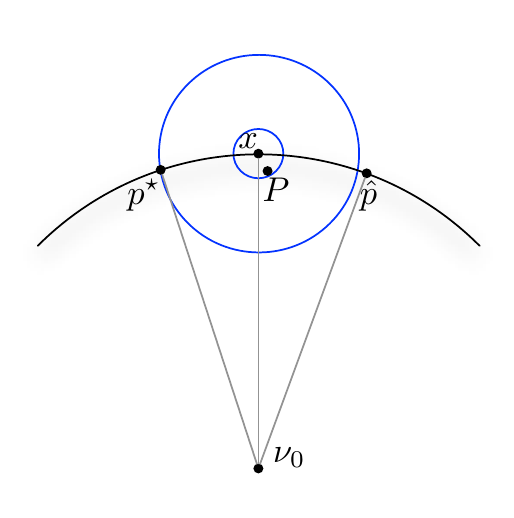}}
\vspace{-5mm}
\caption{\label{fig:kern-pow}
\small \sffamily 
Illustration of $x$, $p^\star$, $\hat p$, $\nu_0$, and $P$ as vectors in a RKHS.  Note we have omitted the $\phi_K$ and $\Phi_K$ maps to unclutter the notation.  
}
\end{figure}

We now assume that $D_K(P,x) < D_K(p^\star,x)/C_\sigma$ and show this is not possible.  
First observe that $D_K(P, \hat p) + D_K(P, x) \geq D_K(\hat p, x) \geq D_K(p^\star, x)$.  
These expressions imply that $D_K(P,\hat p) \geq D_K(p^\star,x) - D_K(P,x) \geq (1-1/C_\sigma)D_K(p^\star, x)$, and thus 
\[
D_K(P,x) 
\geq 
\frac{1}{2 \sigma} D_K(P,\hat p)^2 
\geq 
\frac{1}{2 \sigma} \left(1-\frac{1}{C_\sigma}\right)^2 D_K(p^\star,x)^2
\geq 
\frac{1}{C_\sigma}D_K(p^\star,x)^2,
\]
a contradiction.  
The last steps follows by setting 
\begin{align*}
\frac{1}{2 \sigma} \left(1-\frac{1}{C_\sigma}\right)^2 & \geq \frac{1}{C_\sigma}
\Rightarrow
C_\sigma^2 - (2+2\sigma) C_\sigma + 1  \geq 0
\end{align*}
and solving for $C_\sigma$, 
\begin{align*}
C_\sigma 
 \geq 
\frac{(2+2 \sigma) + \sqrt{(2+2 \sigma)^2 - 4}}{2}
= 
1+ \sigma + \sqrt{\sigma^2 + 2 \sigma}
= 
1 +  \sigma + \sqrt{( \sigma +1)^2 - 1}.
\end{align*}
Since $C_\sigma =  2\sigma + 2 > 1 +  \sigma + \sqrt{( \sigma +1)^2 - 1}$, so we have $\frac{1}{2 \sigma} \left(1-\frac{1}{C_\sigma}\right)^2  \geq \frac{1}{C_\sigma}$.  
\end{proof}

Recall that an $\eps$-kernel sample $P$ of $\mu$ satisfies
$\max_{x \in \Rspace^d} |\kappa(\mu,x) - \kappa(\mu_P,x) | \leq \eps.$

\begin{theorem}
\label{thm:powK-up-eh}
If $D_K(P,x) \geq 1$ then $\Kpow{P}(P,x) \leq \sqrt{6 \sigma + 8} D_K(P,x)$. 
If $P$ is an $(\eps/4)$-kernel sample of $\mu$ then 
$\Kpow{P}(\mu,x) \leq \sqrt{6 \sigma+8} (D_K(\mu,x) +\eps)$.  
\end{theorem}
\begin{proof}
We combine Lemma \ref{lem:sq-bnd} with Lemma \ref{lem:pK-up} to achieve
\[
\Kpow{P}(P,x)^2 
\leq 
2 D_K(P,x)^2 + 3 D_K(p^\star,x)^2
\leq
2 D_K(P,x)^2 + 3(2\sigma+2)  D_K(P,x).
\]
\emph{Aside:  Note that the first $D_K(P,x)$ is squared and the second is not.}
If $D_K(P,x) \geq \alpha$ then $D_K(P,x) \leq (1/\alpha) D_K(P, x)^2$ we have
\[
\Kpow{P}(P,x)^2 \leq (2 + (6 + 6 \sigma)/\alpha) D_K(P,x)^2. 
\]
Let $\alpha=1$.  We have 
\[
\Kpow{P}(P,x)^2 \leq (6 \sigma + 8) D_K(P,x)^2. 
\]

Since $D_K(P,x) \leq D_K(\mu,x)+\eps$, via Lemma \ref{lem:KD-samp}. 
We obtain, 
\[
\Kpow{P}(\mu,x) \leq \sqrt{6 \sigma+8} (D_K(\mu,x) +\eps). \qedhere
\]
\end{proof}

\subsection{Approximating the Minimum Kernel Distance Point}
\label{app:phat+}
The goal in this section is to find a point that approximately minimizes the kernel distance to a point set $P$.  We assume here $P$ contains $n$ points and describes a measure made of $n$ Dirac mass at each $p \in P$ with weight $1/n$ (this is the empirical measure $\muP$ defined in Section \ref{sec:kernel}).  
Let $p_+ = \arg \min_{q \in \Rspace^d} D_K(\muP,q) = \arg \max_{q \in \Rspace^d} \kappa(\muP,q)$. 
Since $D_K(\muP, q) = D_K(P, q)$, for simplicity in notation, we work with point set $P$ instead of $\muP$ for the remaining of this section. 
That is, we define $p_+ = \arg \min_{q \in \Rspace^d} D_K(P,q) = \arg \max_{q \in \Rspace^d} \kappa(P,q)$.  
Note that $p_+$ is chosen over all of $\Rspace^d$, as the bound in Theorem \ref{thm:powK-up-eh} is not sufficient when choosing a point from $P$.  
In particular, for any $\delta>0$, we want a point $\hat p_+$ such that $D_K(P,\hat p_+) \leq (1+\delta) D_K(P,p_+)$.  

Note that Agarwal \etal~\cite{AHKS13} provide an algorithm that with high probability finds a point $\hat q$ such that $\kappa(P,\hat q) \geq (1-\delta) \kappa(P,p_+)$ in time $O((1/\delta^4) n \log n)$.  However this point $\hat q$ is \emph{not} sufficient for our purpose
(that is, $\hat q$ does not satisfy the condition  $D_K(P,\hat q_+) \leq (1+\delta) D_K(P,p_+)$),  since $\hat q$ yields  
\[
D_K(P,\hat q)^2 \leq \sigma^2 + \kappa(P,P) - 2(1-\delta)\kappa(P,p_+) 
\nleq 
(1+\delta)\big(\sigma^2 + \kappa(P,P) - 2\kappa(P,p_+)\big) = (1-\delta) D_K(P,p_+)^2,
\]
since in general it is not true that $4\kappa(P,p_+) \leq \sigma^2 + \kappa(P,P)$, as would be required.  

First we need some structural properties. 
For each point $x \in \Rspace^d$, define a radius $r_x = \arg \sup_{r >0} \{|B_r(x) \cap P| \leq n/2$\}, where $B_r(x)$ is a ball of radius $r$ centered at $x$.  In other words, it is the largest radius such that at most half of points in $P$ are within $B_r(x)$.  
Let $\hat p_2$ be the point in $P$ such that $\|p_+ - \hat p_2\| = r_{p_+}$.  
In other words, $\hat p_2$ is a point such that no more than $n/2$ points in $P$ satisfy $\|p_+ - p\| \geq \|p_+ - \hat p_2\|$.  
Finally it is useful to define $r_{x,K}$ which is $r_{x,K} = D_K(x,p)$ where $\|x-p\| = r_x$; in particular $r_{p_+,K} = D_K(p_+, \hat p_2)$.  

We now need to lower bound $D_K(P,p_+)$ in terms of $D_K(P,\hat p_2)$.
Lemma \ref{lem:sq-bnd} already provides a bound in terms of the closest point for any $x \in \Rspace^d$.  We follow a similar construction here.

\begin{lemma}
\label{lem:vec-shrink2}
Consider a set $V = \{v_1, \ldots, v_n\}$ of vectors in a Hilbert space endowed with norm $\|\cdot\|$ and inner product $\langle \cdot, \cdot \rangle$.  Let each $v_i$ have norm $\|v_i\| = \eta$.  
Consider weights $W = \{w_1, \ldots, w_n\}$ such that $1/2 \geq w_i \geq 0$ and $\sum_{i=1}^n w_i = 1$.  Let $r = \sum_{i=1}^n w_i v_i$.  Define a partition of $V$ with $V_1$ and $V_2$ such that $V_2$ is the smallest set such that $\sum_{v_i \in V_2} w_i \geq 1/2$, and for all $v_1 \in V_1$ and $v_2 \in V_2$ we have $\| r- v_1\| < \| r - v_2\|$.  Let $\hat v_2 = \arg \min_{v_i \in V_2} \|v_i - r\|$.
Then
 \[
 \|r\|^2 \leq \eta^2 - \frac{\|r- \hat v_2\|^2}{2}.
 \]
\end{lemma}
\begin{proof}
For ease of notation, we assume that $\langle v_i,r\rangle > \langle v_{i+1}, r\rangle$ for all $i$, and let $\{v_1, \ldots, v_k\} = V_1$.    Let $\hat v_1 = \arg \min_{v_i \in V_1} \|v_i - r\| = \arg \min_{v_i \in V} \|v_i - r\|$.  
Let $\hat v$ be a norm $\eta$ vector that has $\langle \hat v,  r\rangle = (\langle \hat v_1, r \rangle + \langle \hat v_2, r \rangle)/2$.
Since  $\sum_{v_i \in V_2} w_i \geq 1/2$ and $\sum_{v_i \in V_1} w_i \leq 1/2$, 
let $\sum_{v_i \in V_2} w_i = 1/2 + \delta$ and $\sum_{v_i \in V_1} w_i = 1/2 - \delta$ (for $0 \leq \delta \leq 1/2)$. 
By definition, we also have $\langle \hat v_1, r \rangle \geq \langle \hat v_2, r \rangle$. 
We can decompose $r$ as 
\begin{align*}
\|r\|^2 
= 
\langle r, r\rangle 
&= 
\sum_{i=1}^n w_i \langle v_i, r \rangle
=
\sum_{i=1}^k w_i \langle v_i, r\rangle + \sum_{i={k+1}}^n w_i \langle v_i, r \rangle
\\ &\leq
\sum_{i=1}^k w_i \langle \hat v_1, r\rangle + \sum_{i=k+1}^n w_i \langle \hat v_2, r \rangle
= 
\left(\sum_{i=1}^k w_i\right) \langle \hat v_1, r\rangle  + \left(\sum_{i=k+1}^n w_i \right) \langle \hat v_2, r \rangle 
\\ &
= (1/2 - \delta) \langle \hat v_1, r \rangle + (1/2 + \delta)  \langle \hat v_2, r \rangle 
= (1/2) ( \langle \hat v_1, r\rangle + \langle \hat v_2, r \rangle) + \delta( \langle \hat v_2, r \rangle -  \langle \hat v_1, r \rangle)
\\ & \leq
(\langle \hat v_1, r \rangle + \langle \hat v_2, r \rangle)/2
= 
\langle \hat v, r \rangle
\\ &= 
\frac{1}{2} ( \|r\|^2 + \|\hat v\|^2 - \|\hat v - r\|^2).
\end{align*}
The last inequality holds by $\|\hat v - r\|^2 = \|r\|^2 + \|\hat v\|^2 - 2 \langle \hat v, r \rangle$.  
Then since $\|\hat v\| = \eta$ we can solve for $\|r\|^2$ as 
\[
\|r\|^2 \leq \eta^2 - \|\hat v - r\|^2 = \eta^2 - (\|\hat v_2 -r\|^2 + \|\hat v_1 - r\|^2)/2
\leq
\eta^2 - \|\hat v_2 - r\|^2/2.  \qedhere
\]
\end{proof}

\begin{lemma}
\label{lem:HK-shrink2}
Using $\hat p_2$ as defined above, then
$\kappa(P,P) \leq \sigma^2 - D_K(P,\hat p_2)^2/2$.
\end{lemma}
\begin{proof}
Let $\phi_K : \Rspace^d \to \Eu{H}_K$ map points in $\Rspace^d$ to the reproducing kernel Hilbert space (RKHS) $\Eu{H}_K$ defined by kernel $K$.  This space has norm $\|P\|_{\Eu{H}_K} = \sqrt{\kappa(P,P)}$ defined on a set of points $P$ and inner product $\kappa(P,P)$.  Let $\Phi_K(P) = \frac{1}{|P|}\sum_{p \in P} \phi_K(p)$ be the representation of a set of points $P$ in $\Eu{H}_K$.  
Note that $D_K(P,Q) = \|\Phi_K(P) - \Phi_K(Q)\|_{\Eu{H}_K}$.  
We can now apply Lemma \ref{lem:vec-shrink2} to $\{\phi_K(p)\}_{p \in P}$ with weights $w(p) = 1/|P|$ and $r = \Phi_K(P)$, and norm $\eta = \sigma$.  
Finally note that we can use $\phi_K(\hat p_2) = \hat v_2$ since $V_2$ represents the set of points which are further or equal to $P$ than is $\hat p_2$.   
In addition, by the property of RKHS, $\| \Phi_K(P) - \phi_K(\hat p_2) \| = D_K(P, \hat p_2)$. 
Hence
$\kappa(P,P) = \|P\|_{\Eu{H}_K}^2 \leq \sigma^2 - D_K(P,\hat p_2)^2/2$.
\end{proof}

\begin{lemma}
\label{lem:sq-bnd2}
$D_K(P,p_+) \geq D_K(p_+,\hat p_2)^2/(4\sigma)$.  
\end{lemma}
\begin{proof}
Refer to Figure \ref{fig:kern-pow} for geometric intuition in this proof.  
Let $\nu_0$ be a measure that is $\nu_0(p) = 0$ for all $p \in \Rspace^d$; thus it has a norm $\kappa(\nu_0, \nu_0) = 0$.  
We can measure the distance from $\nu_0$ to $p_+$ and $P$, noting that $D_K(\nu_0,x) = \sqrt{\kappa(x,x)} = \sigma$ and $D_K(P,\nu_0) = \sqrt{\kappa(P,P)}$.  Thus by triangle inequality, Lemma \ref{lem:HK-shrink2}, and Lemma \ref{lem:outa-root} 
\begin{align*}
D_K(P,p_+) 
& \geq 
D_K(\nu_0, p_+) - D_K(P, \nu_0) 
\\ &= 
\sigma - \sqrt{\kappa(P,P)}
\\ &\geq
\sigma - \sqrt{\sigma^2 - D_K(P,\hat p_2)^2/2} 
\\ & \geq
D_K(P,\hat p_2)^2/4\sigma.  \qedhere
\end{align*}
\end{proof}

Now we place a net $\Eu{N}$ on $\Rspace^d$; specifically, it is a set of points such that for some $q \in \Eu{N}$ that
$\|q - p_+\| \leq \delta D_K(p_+,\hat p_2)^2/4\sigma \leq \delta D_K(P,p_+)$ (we refer to this inequality as the \emph{net condition}, therefore, $\Eu{N}$ is a set of points such that some points in it satisfy the net condition).  Since $D_K(P,\cdot)$ is $1$-Lipschitz, we have 
$D_K(P, p_+) - D_K(P, q) \leq \|q - p_+ \|$. 
This ensures that some point $q \in \Eu{N}$ satisfies $D_K(P,q) \leq (1+\delta) D_K(P, p_+)$, and can serve as $\hat p_+$.  
In other words, $\Eu{N}$ is guaranteed to contain some point $q$ that can serve as $p_+$.  

Note that $p_+$ must be in $\CH(P)$, the convex hull of $P$.  Otherwise, moving to the closest point on $\CH(P)$ decreases the distance to all points, and thus increases $\kappa(P,p_+)$, which cannot happen by definition of $p_+$.  Let $\Delta$ be the diameter of $P$ (the distance between the two furthest points).  Clearly for some $p \in P$ we must have $\|p_+ - p\| \leq \Delta$.  

Also note that $p_+ : = \arg \max_{q \in \Rspace^d} \kappa(P, q)$ must be within a distance $R_\sigma = \sigma \sqrt{2 \ln(n)}$ to some $p \in P$, otherwise for $p^\star = \arg \min_{p \in P} \|p_+ -p\|$, we can bound
$
\kappa(P,p_+) \leq K(p^\star,p_+) \leq \sigma^2/n = K(p^\star,p^\star)/n \leq \kappa(P,p^\star),
$ which means $p_+$ is not a maximum.  
The first inequality is by definition of $p^*$, the second by assuming $\|p_+ - p^\star\| \geq \sigma \sqrt{2 \ln(n)}$.

Let $B_R(p)$ be the ball centered at $p$ with radius $R = \min(R_\sigma,\Delta)$.  
Let $R_p = \min(R,r_p/2)$.  
So $p_+$ must be in $\bigcup_{p \in P} B_R(p)$.  
We describe a net $\Eu{N}_p$ construction for one ball $B_R(p)$; that is for any $x$ such that $p \in P$ is the closest point to $x$, then some point $q \in \Eu{N}_p$ satisfies  $\|q-x\| \leq \delta (r_{x,K})^2/4 \sigma$.  Thus if this point $x = p_+$, the correct property holds, and we can use the corresponding $q$ as $\hat p_+$.  
Then $\Eu{N} = \bigcup_{p \in P} \Eu{N}_p$, and is at most $n$ times the size of $\Eu{N}_p$.  
Let $k_p$ be the smallest integer $k$ such that $r_p/2 \geq R/2^k$.  
The net $\Eu{N}_p$ will be composed of $\Eu{N}_p = \bigcup_{i=0}^{k_p} \Eu{N}_i = \Eu{N}_0 \cup \Eu{N}'_p$, where $\Eu{N}'_p = \bigcup_{i=1}^{k_p} \Eu{N}_i$. 

Before we proceed with the construction, we need an assumption:  That $\Lambda_P = \min_{p \in P} r_p$ is a bounded quantity, it is not too small.  That is,  no point has \emph{more} than half the points within an absolute radius $\Lambda_P$.  We call $\Lambda_P$ the \emph{median concentration}.  

\begin{lemma}
\label{lem:N0}
A net $\Eu{N}_0$ can be constructed of size $O((\sigma/ \delta \Lambda_P)^d + \log^{d/2}(n))$ so that all points $x \in B_{R_p}(p)$ satisfy $\|q - x\| \leq \delta (r_{x,K})^2/4 \sigma$ for some $q \in \Eu{N}_0$.
\end{lemma}
If $x = p_+$, then such a point satisfies the net condition, that is there is a point $q \in \Eu{N}_0$ such that   
$\|q - x\| = \| q - p_+\| \leq \delta (r_{p_+,K})^2/(4 \sigma) = \delta D_K(p_+,\hat p_2)/(4\sigma) \leq \delta D_K(P, p_+)$.  
\begin{proof}
For all points $x \in B_{R_p} \subset B_{r_p/2}(p)$, they must have $r_x \geq r_p/2$, otherwise $B_{r_p/2}(x)$ is completely inside $B_{r_p}(p)$, and cannot have enough points.   
Within $B_{R_p}(p)$ we place the net $\Eu{N}_0$ so that all points $x \in B_{R_p}(p)$ satisfy $\|x-q\| \leq \min(\delta r_p^2/32\sigma, \sqrt{3}\sigma)$ for some $q \in \Eu{N}_0$.  
Now $\delta r_p^2/32 \sigma \leq \delta r_x^2 / 8 \sigma$, and since $\|x-y\|^2/2 \leq D_K(x,y)^2$ (for $\|x-y\| \leq \sqrt{3}\sigma$, via Lemma \ref{lem:E2DK}), thus the net ensures if $p_+ \in B_{R_p}(p)$, then some $q \in \Eu{N}_0$ is sufficiently close to $p_+$.  

Since $B_{R_p}(p)$ fits in a squared box of side length $\min(2 R_\sigma,r_p)$, then we can describe $\Eu{N}_0$ as an axis-aligned grid with $g$ points along each axis.  We define two cases to bound $g$.  When $\delta r_p^2/32 \sigma < \sqrt{3} \sigma$ then we can set
\[
g 
= 
\frac{R_p}{\delta r_p^2/ (32\sigma\sqrt{d}) } 
\leq 
\frac{32 \sigma \sqrt{d} }{\delta r_p} 
= 
O(\sigma/\delta r_p) = O(\sigma / \delta \Lambda_P)
\]
Otherwise, 
\[
g = 
\frac{R_p}{\sqrt{3} \sigma/\sqrt{d}}
\leq
\frac{\sigma \sqrt{2\ln(n)}}{\sqrt{3} \sigma/\sqrt{d}}
=
\sqrt{2d \ln(n)/3} = O(\sqrt{\log(n)}).
\]
Then we need $|\Eu{N}_0| = O(g^d) = O((\sigma / \delta \Lambda_P)^2 + \ln^{d/2}(n))$.  
\end{proof}

When $r_p/2 < R$ we still need to handle the case for $x \in A_p$ where the annulus $A_p = B_R(p) \setminus B_{r_p/2}(p)$.  
For a point $x \in A_p$ if $p = \min_{p' \in P} \|x-p'\|$ then $r_x \geq \|x-p\|$.  We only worry about the net $\Eu{N}_p'$ on $A_p$ for these points where $p$ is the closest point, the others will be handled by another $\Eu{N}_{p'}$ for $p' \in P$ and $p' \neq p$.  

Recall $k_p$ is the smallest integer $k$ such that $r_p/2 \geq R/2^k$.  

\begin{lemma}
\label{lem:Ann-net}
A net $\Eu{N}_p'$ can be constructed of size $O(k_p + (\sigma/ \delta \Lambda_P)^d + \log^{d/2}(n))$ so that all points $x \in A_p$ where $p = \arg \min_{p' \in P} \|x-p'\|$,  satisfy $\|q - x\| \leq \delta (r_{x,K})^2/4 \sigma$ for some $q \in \Eu{N}_p'$.  
\end{lemma}

If $x = p_+$, then such a point satisfies the net condition, that is there is a point $q \in \Eu{N}'_p$ such that   
$\|q - x\| = \| q - p_+\| \leq \delta (r_{p_+,K})^2/(4 \sigma) = \delta D_K(p_+,\hat p_2)/(4\sigma) \leq \delta D_K(P, p_+)$.

\begin{proof}
We now consider the $k_p$ annuli $\{A_1, \ldots, A_{k_p}\}$ which cover $A_p$.  Each $A_i = \{x \in \Rspace^d \mid R/2^{i-1} \geq \|p-x\| > R/2^i\}$ has volume $O((R/2^{i-1})^d)$.  For any $x \in A_i$ we have $r_x \geq \|x-p\| \geq R/2^i$, so the Euclidean distance to the nearest $q \in \Eu{N}_i$ can be at most $\min(\sqrt{3}\sigma , \delta (R/2^i)^2/ 8\sigma)$. Thus we can cover $A_i$ with a net $\Eu{N}_i$ of size $t_i$ based on two cases again.  If $\delta (R/2_i)^2/8\sigma < \sqrt{3} \sigma$ then
\[
t_i 
= 
O\left(1 +  \left( \frac{R}{2^i} / \left(\frac{\delta}{\sigma} \left(\frac{R}{2^i}\right)^2\right)\right)^d\right)
=
O\left(1 + \left( \frac{2^i}{R} \frac{\sigma}{\delta}\right)^d\right)
=
O(1) + O\left( \left(\frac{\sigma}{\delta R}\right)^d (2^d)^i\right).
\]
Otherwise
\[
t_i 
= 
O\left( 1+ \left(\left(\frac{R}{2^i}\right) / \sqrt{3}\sigma \right)^d \right)
=
O\left( 1+ \left(\frac{R_\sigma = \sigma \sqrt{\log(n)}}{2^i \sigma }\right)^d \right)
=
O(1) + O\left(\frac{\log^{d/2}(n)}{(2^d)^i}\right).
\]
Since $R/2^{k_p} \geq r_p/2 \geq \Lambda_P/2$, then
the total size of $\Eu{N}_p'$, the union of all of these nets, is $\sum_{i=1}^{k_p} t_i \leq O(k_p) + 2t_{k_p}  + 2 t_1 = O(k_p + (\sigma/\delta \Lambda_P)^d + \log^{d/2}(n))$.  
In the first case $t_{k_p}$ dominates the cost and in the second case it is $t_1$.  
\end{proof}

Thus the total size of $\Eu{N}_p$ is $O((\sigma/\delta \Lambda_P)^d + \log^{d/2}(n) + k_p)$ where $k_p \leq \log(R/r_p)+2$.  
It just remains to bound $k_p$.  Given that no more than $n/2$ points are collocated on the same spot (which already holds by $\Lambda_P$ being a bounded quantity), then for all $p \in P$, $r_p \geq \min_{q \neq q' \in P} \|q-q'\|$.  
The value $\beta_P = \Delta / \min_{q \neq q' \in P} \|q-q'\|$ is known as the \emph{spread} of a point set, and it is common to assume it is an absolute bounded quantity related to the precision of coordinates, where $\log(\beta_P)$ is not too large.  Thus we can bound $k_p = O(\log(\beta_P))$.  

\begin{theorem}
\label{thm:phat+}
Consider a point set $P \subset \Rspace^d$ with $n$ points, spread $\beta_P$, and median concentration $\Lambda_P$.  
For any $\delta > 0$, in time $O(n^2 ((\sigma/\delta \Lambda_P)^d + \log^{d/2}(n)+ \log(\beta_P)))$ we can find a point $\hat p_+$ such that $D_K(P,\hat p_+) \leq (1 + \delta) D_K(P,p_+)$.  
\end{theorem}
\begin{proof}
Using Lemma \ref{lem:N0} and Lemma \ref{lem:Ann-net} we can build a net $\Eu{N}$ of size $O(n((\sigma/\delta \Lambda_P)^d + \log^{d/2}(n) + \log(\beta_P))$ such that some $q \in \Eu{N}$ satisfies $\|q-p_+\| \leq \delta D_K(q,p_+)^2 /4 \sigma \leq \delta D_K(P,p_+)$.  Lemma \ref{lem:sq-bnd2} ensures that this $q$ satisfies $D_K(P,q) \leq (1 + \delta) D_K(P,p_+)$ since $D_K(P,\cdot)$ is $1$-Lipschitz.  

We can find such a $q$ and set it as $p_+$ by evaluating $\kappa(P,q)$ for all $q \in \Eu{N}$ and taking the one with largest value.  This takes $O(n)$ for each $q \in \Eu{N}$. 
\end{proof}

We claim that in many realistic settings $\sigma/\Lambda_P = O(1)$.  In such a case the algorithm runs in $O(n^2 (1/\delta^d + \log^{d/2} n + \log(\beta_P)))$ time.  
If $\sigma/\Lambda_P = o(1)$, then \emph{over} half of the measure described by $P$ will essentially behave as a single point.  In many settings $P$ is drawn uniformly from a compact set $S$, so then choosing $\sigma$ so that more than half of $S$ has negligible diameter compared to $\sigma$ will cause that data to be over smoothed.  
In fact, the definition of $\Lambda_P$ can be modified so that this radius never contains more than any $\tau n$ points for any constant $\tau < 1$, and the bounds do not change asymptotically.

\section{Details on Reconstruction Properties of Kernel Distance}
\label{app:recon}

In this section we provide the full proof for some statements from Section \ref{sec:recon}.  

\subsection{Topological Estimates using Kernel Power Distance}
\label{subsec:approximation}
For persistence diagrams of sublevel sets filtration of $\dK_{\mu}$ and the weighted Rips filtration $\{R_{\alpha}(P,\dK_\mu)\}$ to be well-defined, we need the technical condition (proved in Lemma \ref{lemma:q-tame} and \ref{lemma:q-tame-rips}) that they are $q$-tame. 
Recall a filtration $F$ is \emph{$q$-tame} if for any $\alpha < \beta$, the homomorphism between 
$\Hgroup(F_{\alpha})$ and $\Hgroup(F_{\beta})$ induced by the canonical inclusion has finite rank \cite{ChazalCohen-SteinerGlisse2009,ChazalSilvaGlisse2013}. 

\begin{lemma}
The sublevel sets filtration of $\dK_\mu$ is $q$-tame. 
\label{lemma:q-tame}
\end{lemma}
\begin{proof}
The proof resembles the proof of $q$-tameness for distance to measure sublevel sets filtration (Proposition 12, \cite{BuchetChazalOudot2013}). 
We have shown that $\dK_\mu$ is $1$-Lipschitz and proper. Its properness property implies that any sublevel set $A:=(\dK_\mu)^{-1}([0,\alpha])$ (for $\alpha < c_\mu$) is compact. 
Since $\Rspace^d$ is triangulable (i.e. homeomorphic to a locally finite simplicial complex), there exists a homeomorphism $h$ from $\Rspace^d$ to a locally finite simplicial complex $C$. 
For any $\alpha > 0$, since $A$ is compact, we consider the restriction of $C$ to a finite simplicial complex $C_{\alpha}$ that contains $h(A)$. 
The function $(\dK_\mu \circ h^{-1})\mid_{C_{\alpha}}$ is continuous on $C_{\alpha}$, 
therefore its sublevel set filtration is $q$-tame based on 
Theorem 2.22 of \cite{ChazalSilvaGlisse2013}, which states that the sublevel sets filtration of a continuous function (defined on a realization of a finite simplicial complex) is $q$-tame. 
Extending the above construction to any $\alpha$, the sublevel sets filtration of $\dK_\mu \circ h^{-1}$ is therefore $q$-tame. 
As homology is preserved by homeomorphisms $h$, this implies that the sublevel sets filtration of $\dK_{\mu}$ is $q$-tame. 
\end{proof}

Setting $\mu = \muP$, Lemma \ref{lemma:q-tame} implies that the sublevel sets filtration of $\dK_\muP$ is also $q$-tame. 

\begin{lemma}
The weighted Rips filtration $\{R_{\alpha}(P,\dK_\mu)\}$ is $q$-tame for compact subset $P \subset \Rspace^d$. 
\label{lemma:q-tame-rips}
\end{lemma}
\begin{proof}
Since $P$ is compact subset of $\Rspace^d$, $\dgmD{\{R_{\alpha}(P,\dK_\mu)\}})$ is $q$-tame based on Proposition 32 of \cite{ChazalSilvaGlisse2013}, which states that the weighted Rips filtration with respect to a compact subset $P$ in metric space and its corresponding weight function is $q$-tame. 
\end{proof}

Setting $P = \hP_+$, $\mu = \muP$, Lemma \ref{lemma:q-tame-rips} implies that the weighted Rips filtration $\{R_{\alpha}(\hP_+,\dK_\muP)\}$ is well-defined.

\subsection{Inference of Compact Set $S$ with the Kernel Distance}
\label{app:infer}

Suppose $\mu$ is a uniform measure on a compact set $S$ in $\Rspace^d$. 
We now compare the kernel distance $\dK_\mu$ with the distance function $f_S$ to the support $S$ of $\mu$. We show how $\dK_\mu$ approximates $f_S$, and thus allows one to infer geometric properties of $S$ from samples from $\mu$. 

For a point $x \in \Rspace^d$, the distance function $f_S$ measures the minimum distance between $x$ and any point in $S$, $f_S(x) = \inf_{y \in S}||x-y||$. 
The point $x_S$ that realizes the minimum in the definition of $f_S(x)$ is the \emph{orthogonal projection} of $x$ on $S$.  The location of the points $x \in \Rspace^d$ that have more than one projection on $S$ is the \emph{medial axis} of $S$ \cite{Merigot2010}, denoted as $\ma(S)$. Since $\ma(S)$ resides in the unbounded component $\Rspace^d \setminus S$, it is referred to as the \emph{outer medial axis} similar to the concept found in \cite{Dey2007}. The \emph{reach} of $S$ is the minimum distance between a point in $S$ and a point in its medial axis, denoted as $\reach(S)$. 
Similarly, one could define the medial axis of $\Rspace^d \setminus S$ (i.e. 
the \emph{inner medial axis} which resides in the interior of $S$) following definitions in \cite{Lieutier2004}, and denote its associated reach as $\reach(\Rspace^d \setminus S)$. The concepts of reach associated with the inner and outer medial axis of $S$ capture curvature information of the compact set.

Recall that a generalized gradient and its corresponding flow to a distance function are described in \cite{ChazalCohen-SteinerLieutier2009} and later adapted for distance-like functions in \cite{ChazalCohen-SteinerMerigot2011}. 
Let $f_S: \Rspace^d \to \Rspace$ be a distance function associated with a compact set $S$ of $\Rspace^d$. 
It is not differentiable on the medial axis of $S$. 
It is possible to define a \emph{generalized gradient function} $\grad{S}{}: \Rspace^d \to \Rspace^d$ coincides with the usual gradient of $f_S$ where $f_S$ is differentiable, and is defined everywhere and can be integrated into a continuous flow $\Phi^t: \Rspace^d \to \Rspace^d$.   
Such a flow points away from $S$, towards local maxima of $f_S$ (that belong to the medial axis of $S$) \cite{Merigot2010}. 
The integral (flow) line $\gamma$ of this flow starting at point in $\Rspace^d$ can be parameterized by arc length, $\gamma: [a, b] \to \Rspace^d$, and we have $f_S(\gamma(b)) = f_S(\gamma(a)) + \int_{a}^{b}||\grad{S}{}(\gamma(t))|| d_t$.

\begin{figure}
\includegraphics[width=0.3 \linewidth]{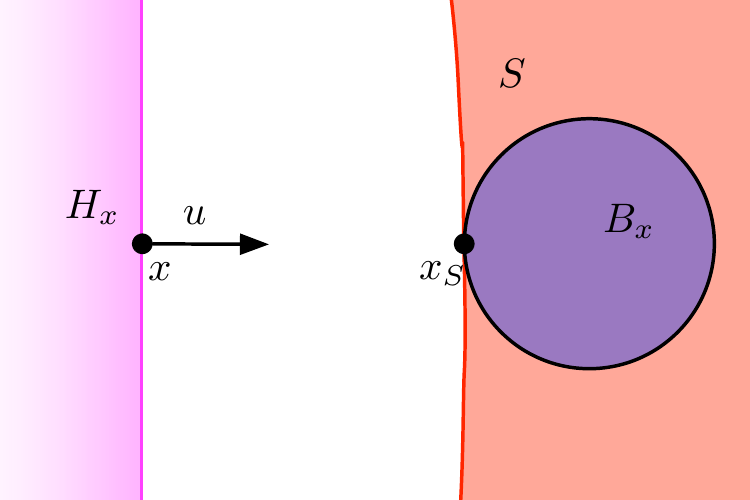}
\hspace{.2in}
\includegraphics[width=0.3\linewidth]{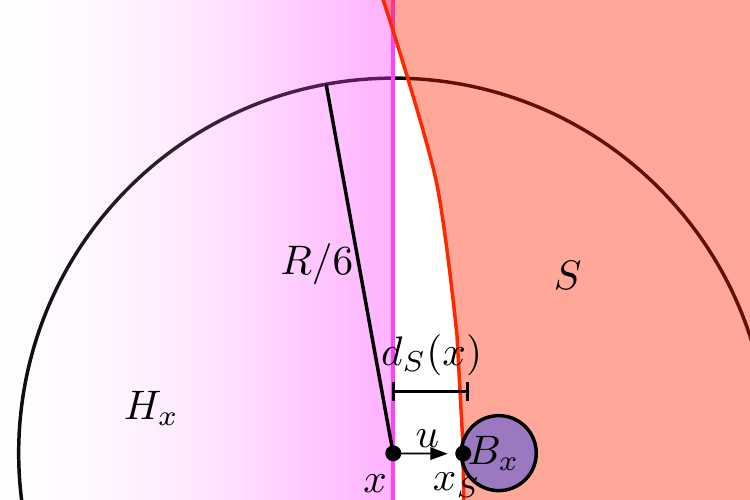}
\hspace{.2in}
\includegraphics[width=0.3 \linewidth]{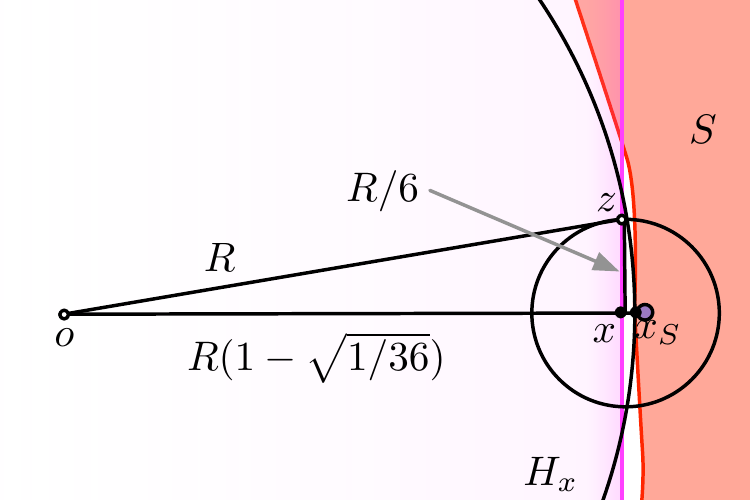}
\caption{\label{fig:inference}
Illustrations of the geometric inference of $S$ from $\dK_\mu$ at three scales.}
\end{figure}

\begin{lemma}[Lemma \ref{lem:monotonicity}]
Given any flow line $\gamma$ associated with the generalized gradient function $\grad{S}{}$, $\dK_\mu(x)$ is strictly monotonically increasing along $\gamma$ for $x$ sufficiently far away from the medial axis of $S$, for $\sigma \leq \frac{R}{6\Delta_G}$ and $ f_S(x) \in (0.014 R, 2\sigma)$. 
Here $B(\sigma/2)$ denotes a ball of radius $\sigma/2$,  $G := \frac{\vol(B(\sigma/2))}{\vol(S)}$, $\Delta_G : = \sqrt{12 + 3\ln(4/G)}$ and suppose $R := \min(\reach(S), \reach(\Rspace^d \setminus S)) > 0$.  
\label{lem:monotonicity-app}
\end{lemma}
\begin{proof}
Since $\dK_\mu(x)$ is always positive, and $\dK_\mu(x) = \sqrt{c_\mu - 2\kde_\mu(x)}$ where $c_\mu$ is a constant that depends only on $\mu$, $K$, and $\sigma$, then it is sufficient to show that $\kde_\mu(x)$ is strictly monotonically decreasing along $\gamma$.  

Let $u$ be the negative of the direction of the flow line $\gamma$ at $x$ (i.e $u$ is a unit vector that points towards $S$). We show that $\kde_\mu(x)$ is strictly monotonically increasing along $u$. 
Informally, we will observe that all parts of $S$ that are ``close'' to $x$ are in the direction $u$, and that these parts dominate the gradient of $\kde_\mu(x)$ along $u$.  
We now make this more formal by describing two quantities, $B_x$ and $A_x$, illustrated in Figure \ref{fig:inference}.  

For a point $x \in \Rspace^d$, let $x_S = \arg \min_{x' \in S} \|x'-x\|$; since $x$ is not on the medial axis of $S$, $x_S$ is uniquely defined and $u$ points in the direction of $(x_S - x)/f_S(x)$. 
First, we claim that there exists a ball $B_x$ of radius $\sigma/2$ incident to $x_S$ that is completely contained in $S$.
This holds since $\sigma/2 \leq \frac{R}{6\Delta_G} < R \leq \reach(\Rspace^d \setminus S)$. 
In addition, since $f_S(x) < 2\sigma$, no part in $B_x$ is further than $3 \sigma$ from $x$.
Second, we claim that no part of $S$ within $\Delta_G \cdot \sigma$ ($\leq R/6$)  of $x$ (this includes $B_x$) is in the halfspace $H_x$ with boundary passing through $x$ and outward normal defined by $u$.  
To see this, let $o$ be the center of a ball with radius $R$ that is incident to $x_S$ but not in $S$, refer to such a ball as $B_o$. This implies that points $o$, $x$ and $x_S$ are colinear. Then a ball centered at $x$ with radius $R/6$ should intersect $S$ outside of $B_o$, and in the worst case, on the boundary of $H_x$. 
This holds as long as $\|x-x_S\| \geq 0.014R \geq (1-\sqrt{35/36})R$; see Figure \ref{fig:inference}.  
Define $A_x = \{y \in S \mid \|x-y\| > \Delta_G \cdot \sigma\}$.

Now we examine the contributions to the directional derivative of $\kde_\mu(x)$ along the direction of $u$ from points in $B_x$ and $A_x$, respectively. 
Such a directional derivative is denoted as $\dd_u \kde_\mu(x)$. 
Recall $\kde_\mu(x) = \int_{y \in S} K(x,y) \dir{\mu(y)}$ and $\mu$ is a uniform measure on $S$, $\dd_u \kde_\mu(x) = \frac{1}{\vol(S)} \int_{y \in S} \dd_u K(x,y)$. 
For any point $y \in \mathbb{R}^d$, we define  
$g(y) := \dd_u K(x,y) = \exp(-\|x-y\|^2/2\sigma^2) \langle y-x,u\rangle$. 
Therefore $\dd_u \kde_\mu(x) = \frac{1}{\textsf{vol}(S)} \int_{y \in S}  g(y)$. 

We now examine the contribution to $\dd_u \kde_\mu(x)$ from points in $B_x$, 
$\frac{1}{\vol(S)}\int_{y \in B_x} g(y)$.  First, for all points $y \in B_x$, since $||x - y|| \leq 3\sigma$, we have $\exp(-\|x-y\|^2/2\sigma^2) \geq \exp(-9/2)$. 
Second, at least half of points $y \in B_x$ (that covers half the volume of $B_x$) is at least $\sigma/2$ away from $x_S$, and correspondingly for these points $\langle y-x,u\rangle \geq \sigma/2$. 
We have $\int_{y \in B_x} g(y) \geq \frac{1}{2} \vol(B_x) \cdot \exp(-9/2) \cdot \sigma/2$. Given $\vol(B_x) = G \cdot \vol(S)$, we have $\frac{1}{\vol(S)}\int_{y \in B_x} g(y) \geq  \frac{1}{4} G \cdot \exp(-9/2) \cdot \sigma$.   
Denote $B =  \frac{1}{4} G \cdot \exp(-9/2) \cdot \sigma$. 

We now examine the contribution to $\dd_u \kde_\mu(x)$ from points in $A_x$, 
$\frac{1}{\vol(S)}\int_{y \in A_x} g(y)$. For any point $y \in \Rspace^d$ (including $y \in A_x$), 
$\langle y-x,u\rangle \leq \|x-y\|$. 
Let $\phi_y = \|x-y\|/\sigma$ so we have $g(y) \leq \exp(-\phi_y^2/2) \phi_y \sigma$. 
Since this bound on $g(y)$ is maximized at $\phi_y = 1$, under the condition $\phi_y \geq \Delta_G \geq \sqrt{12}>1$, we can set $\phi_y=\Delta_G$ to achieve the bound $g(y) \leq \exp(-\Delta_G^2/2) \cdot \Delta_G \sigma$ for $\|x-y\| \geq \Delta_G \cdot \sigma$ (that is, for all $y \in A_x$).  
Now we have $\int_{y \in A_x} g(y) \leq \vol(S) \exp(-\Delta_G^2/2) \cdot \Delta_G \sigma$, leading to $\frac{1}{\vol(S)} \int_{y \in A_x} g(y) \leq  \exp(-\Delta_G^2/2) \cdot \Delta_G \sigma$. Denote $A = \exp(-\Delta_G^2/2) \cdot \Delta_G \sigma$. 

Since only the points $y \in A_x$ could possibly reside in $H_x$ and thus can cause $g(y)$ to be negative, we just need to show that $B > A$.  
This can be confirmed by plugging in $\Delta_G = \sqrt{12 + 3\ln(4/G)}$, and using some algebraic manipulation.  
\end{proof}


\section{Lower Bound on Wasserstein Distance}
\label{app:stability}

We note the next result is a known lower bounds for the Earth movers distance~\cite{Coh99}[Theorem 7].  We reprove it here for completeness.  

\begin{lemma}[\textbf{Lemma.} \ref{lem:EMD}]
For any probability measures $\mu$ and $\nu$ defined on $\Rspace^d$ we have
$\|\bar \mu - \bar \nu\| \leq W_2(\mu,\nu).$
\end{lemma}
\begin{proof}
Let $\pi : \Rspace^d \times \Rspace^d \to \Rspace^+$ describes the optimal transportation plan from $\mu$ to $\nu$.  Also let $u_{\mu,\nu} = \frac{(\bar \mu - \bar \nu)}{\|\bar \mu - \bar \nu \|}$ be the unit vector from $\bar \mu$ to $\bar \nu$.   Then we can expand 
\begin{align*}
(W_2(\mu,\nu))^2
= 
\int_{(p,q)} \|p-q\|^2 \dir{\pi(p,q)}
& \geq
\int_{(p,q)} (\langle (p-q), u_{\mu,\nu} \rangle)^2 \dir{\pi(p,q)}
\\ & \geq
\| \bar \mu - \bar \nu\|^2.
\end{align*}
The first inequality follows since $\langle (p-q), u_{\mu,\nu} \rangle$ is the length of a projection and thus must be at most $\|p-q\|$.  
The second inequality follows since that projection describes the squared length of mass $\pi(p,q)$ along the direction between the two centers $\bar \mu$ and $\bar \nu$, and the total sum of squared length of unit mass moved is exactly $\|\bar \mu - \bar \nu\|^2$.  Note the left-hand-side of the second inequality could be larger since some movement may cancel out (e.g. a rotation).  
\end{proof}

\end{document}